\documentclass[a4paper,11pt]{article}
\pdfoutput=1 

\usepackage{jheppub} 

\usepackage[utf8]{inputenc} 

\usepackage{amsmath}
\usepackage{amssymb}
\usepackage{amsthm}
\usepackage{mathrsfs}
\usepackage{comment}
\usepackage{graphicx}

\usepackage[toc,page]{appendix}

\usepackage{physics}
\usepackage{braket}

\newcommand{\C}{\mathbb{C}}
\newcommand{\R}{\mathbb{R}}

\newtheorem{definition}{Definition}
\newtheorem{proposition}{Proposition}

\newtheorem{remark}{Remark}

\usepackage{xcolor}

\title{Quantum Hall States response to toroidal geometry deformation}

\author[b,d]{Bruno Mera}
\author[a,b]{José M. Mourão}
\author[a,b]{João P. Nunes}
\author[c]{Carolina Paiva}
\affiliation[a]{Center for Mathematical Analysis, Geometry and Dynamical Systems, Instituto Superior Técnico,
Universidade de Lisboa, 1049-001 Lisboa, Portugal}
\affiliation[b]{Department of Mathematics, Instituto Superior Técnico, Universidade de Lisboa, 1049-001 Lisboa,
Portugal}
\affiliation[c]{Raymond and Beverly Sackler School of Physics and Astronomy, Tel Aviv University, Tel Aviv 6997801, Israel}
\affiliation[d]{Instituto de Telecomunicações, 1049-001 Lisboa, Portugal}

\emailAdd{bruno.mera@tecnico.ulisboa.pt}
\emailAdd{paiva@mail.tau.ac.il}
\emailAdd{jpnunes@math.tecnico.ulisboa.pt}
\emailAdd{jmourao@tecnico.ulisboa.pt}

\makeatletter
\gdef\@fpheader{}
\makeatother

\begin{document} 

\abstract{In this paper, we apply techniques of geometric quantization to study the response of the integer and fractional quantum Hall effects to toroidal geometry deformation. 
The main method is that of using complex time Hamiltonian evolution
to induce the geometry change
and then the associated generalized coherent state transforms (gCST) to find 
the evolution of the Laughlin states.

We consider two kinds of deformations. The first 
are flat toroidal deformations.
Although Laughlin states for all flat toroidal geometries have been thoroughly
 studied
before, we believe that our approach via the 
gCST is novel.  It also serves as a  testing ground to study
the non-flat K\"ahler deformations.
The Hamiltonians used in the flat deformations are quadratic in the generators of translations and therefore non periodic.

The second kind of deformations involve nonflat K\"ahler toroidal deformations, generated by global, thus bi-periodic,  Hamiltonians 
on the torus. 
The corresponding imaginary time flows are (elliptic curve modulus) $\tau$--preserving Mabuchi geodesics
in the space of K\"ahler metrics on the torus, hitting 
a curvature singularity in finite imaginary time.
By restricting to $S^1$--invariant deformations we find explicit analytic expressions for the
evolution of the toroidal geometry and of the Laughlin states
 all the way to the singularity.

}

\maketitle

{}


\section{Introduction}
\label{sect:intro}  

The quantum Hall effect (QHE) is a cornerstone of modern condensed matter physics, characterized by quantized Hall conductance that is remarkably robust against disorder and interactions. The integer QHE (IQHE) established the connection between transport properties and topological invariants in band theory—the Chern numbers—while the fractional QHE (FQHE) revealed genuinely interacting topological phases exhibiting fractionally charged anyonic excitations and topological ground-state degeneracy. In Laughlin’s construction, the holomorphic structure plays a central role where model ground state wavefunctions (e.g., the Laughlin state at filling $\nu=1/q$) appear as holomorphic sections with Gaussian weight.

In this geometric setting, the magnetic field on the two-dimensional sample surface acts as an effective symplectic form. Together with the Riemannian metric, it determines a compatible complex structure, thus endowing the surface with a Kähler geometry. This observation allows the use of tools from geometric quantization—particularly holomorphic (or Kähler) quantization—in describing quantum states of the QHE~\cite{Woodhouse}. These techniques have proved particularly effective in the analysis of Laughlin states, enabling the study of their evolution under deformations of the underlying Riemannian metric. Of particular interest are deformations corresponding to geodesic rays in the space of Kähler metrics, generated by the analytic continuation to imaginary time of Hamiltonian flows on the surface \cite{donaldson:1999,MouraoPimentelComplextime}. The evolution of quantum states along these families of Kähler quantizations can be described analytically in terms of generalized coherent state transforms (gCST) (See, for example, \cite{MouraoPimentelGCST} for the application in the study of quantizations of $T^*K$ where $K$ is a compact Lie group). These methods were recently applied in \cite{matos:mera:mourao:mourao:nunes:2023} to study Laughlin states of the FQHE on the two-sphere. Starting from the known expression for the Laughlin state on the round sphere \cite{Haldane2}, the authors obtained explicit deformations along axially symmetric metrics. Remarkably, the particle density exhibited nontrivial behavior as the sphere was deformed into a long, thin cigar while preserving total area.

In the present paper, we extend these geometric quantization techniques to flat and non-flat geometries on the two-torus. In Section \ref{sec:flat_geometry}, we show that the gCST reproduces the known Laughlin states~\cite{Haldane-Rezayi, Murayama, Johri.Papic.Schmitteckert.Bhatt.Haldane2016} for various flat geometries on the torus, while in the Appendix \ref{section:flatgeometriesplane}, we confirm the same for the plane. These results provide further evidence of the validity of the geometric quantization approach and its consistency with known constructions.

In Section \ref{sec:non_flat_geometry}, we apply the gCST to the study of Laughlin states for non-flat geometries on the two-torus, obtained through deformations of the Kähler structure generated by imaginary-time Hamiltonian flows. In these cases—much like for the deformations of the sphere in \cite{matos:mera:mourao:mourao:nunes:2023} and of the plane discussed above—the Hamiltonian is global function on the compact torus, and the complex structure varies within a fixed equivalence class under biholomorphisms, while the symplectic form remains fixed. The result is an deformation of the Riemannian metric, but not of the complex manifold itself. By contrast, in the flat torus case of Section \ref{sec:flat_geometry}, the use of a Hamiltonian function that is smooth on the universal cover but non-periodic—hence not globally defined on the torus—permits an effective variation of the modular parameter. This corresponds to a genuine change of complex structure class, rather than a  deformation within one. From the standpoint of Kähler geometry and geometric quantization, this mechanism is particularly interesting and suggests possible extensions to more general symplectic manifolds.

\section{Preliminaries}
\label{sect:Preliminary_Concepts}

In this section, we will first review the use of Hamiltonian flows in imaginary time in the deformation of Kähler structures and their application to geometric quantization. Then we review some basics of theta functions. These results  give the necessary background for the description of one-particle and many-particle Laughlin states for 
the quantum Hall effect on the torus.

\subsection{Complex symplectomorphisms and K\"ahler structure deformations}
\label{sub:Hamiltonian_Flows_Imaginary_Time}

Here, we will describe deformations of the geometry of Kähler manifolds which, in the so-called symplectic picture, correspond to changing the complex structure while keeping the symplectic form fixed. In order to do that, we need first to introduce a few results from the theory of flows of complex Hamiltonians. 

Let $(M, \omega, J, \gamma)$ be a compact Kähler manifold, where $\omega$ is the symplectic structure, $J$ is the complex structure and $\gamma$ is the Riemannian metric and where we assume that the three structures are 
real-analytic. Let $H$ be a real-valued real-analytic Hamiltonian function on $M$ and $X_{H}$ the associated Hamiltonian vector field, i.e. $\iota_{X_H}\omega = dH$. From the theory of ODE's it follows that the flow of the vector field $X_{H}$, $\varphi^{X_H}_{t}$ is also real-analytic in $t$. Let $C^{\omega}(M)$ be the space of real-analytic functions on $M$. From the theory of Lie series (see \cite{MouraoPimentelComplextime} and references therein), if $f\in C^{\omega}(M)$ there exists $T>0$ such that for $\tau\in D_{0}=\{\tau \in \mathbb{C}:|\tau|<T\}$ we have that the Lie series

\begin{equation}
\left(\varphi^{X_H}_{\tau}\right)^{*}f = e^{\tau X_{H}}\cdot f = \sum_{k=0}^{\infty} \frac{X_{H}^k(f)}{k!}\tau^k.
\label{eq:Hamiltonian_flow}
\end{equation}
is absolutely and uniformly convergent on compact subsets in $M\times D_{0}$.

In particular, one can apply Eq.~\eqref{eq:Hamiltonian_flow} to local $J$-holomorphic coordinates in $M$, such that for small enough $\tau$, one obtains a new system of local complex holomorphic coordinates

\begin{equation}
\label{eq:zitau}
z^{i}_{\tau}=e^{\tau X_{H}}\cdot z^{i}.
\end{equation}

These new local systems of holomorphic coordinates glue together to produce a new globally well defined complex structure in $M$, $J_{\tau}$ ~\cite{MouraoPimentelComplextime}. The local coordinate systems $z^{i}_{\tau}$, as defined above, are local $J_{\tau}$-holomorphic coordinates. In fact, this construction defines a unique biholomorphism:

\begin{equation}
   \varphi^{X_H}_{\tau}:(M,J)\rightarrow (M,J_{\tau}).
   \label{eq:biholomorphism_complexstructure}
\end{equation}

The fact that that $J_\tau$ is still compatible with the symplectic structure $\omega$ allows us to  
obtain a new global Kähler structure $(M,\omega,J_{\tau},\gamma_{\tau})$, where the original symplectic form is unchanged ~\cite{MouraoPimentelComplextime}. 
It is then clear that the analytic continuation of Hamiltonian flows to complex time gives  a systematic way of deforming Kähler structures.

\begin{remark}
Notice that in (\ref{eq:Hamiltonian_flow}) and in the definition of $J_\tau$ and of
$\varphi^{X_H}_{\tau}$ in (\ref{eq:zitau}), (\ref{eq:biholomorphism_complexstructure}),
the Hamiltonian $H$ can be complex valued.

\end{remark}

\subsection{Geometric Quantization}
\label{sub:geometric_quantization}

Quantization is the process of constructing a quantum theory from a given classical system, whose phase space is represented by a symplectic manifold $N$ of dimension $2n$. For a general symplectic manifold, for instance in the absence of symmetries, the quantization procedure can be quite intricate, for example due to absence of a natural choice of local Darboux coordinate systems. 

Geometric quantization is a framework which aims at assigning a Hilbert space of quantum states, 
together with appropriate operators acting on it, corresponding to the quantization of a symplectic manifold 
$(N,\omega)$.  Given a symplectic manifold $(N,\omega)$ that satisfies the integrability condition $\left[\tfrac{\omega}{2\pi \hbar}\right]\in H^{2}(M,\mathbb{Z})$, one can construct a line bundle $L$ over $N$ along with a connection $\nabla$ on $L$ which has a curvature $-i\omega/\hbar$. $L$ is called the prequantum line bundle. One is then able to define operators acting on sections of $L$, the prequantum operators of Kostant-Souriau, which safisfy the desired relation between commutators and Poisson brackets.  The prequantum operators act naturally on the space of smooth sections of $L$.

To obtain a Hilbert space of quantum states two other ingredients are necessary: polarizations and half-forms. Polarizations come about because the representation of prequantum operators is not irreducible (see \cite{BrianHall}) making the space of all smooth sections of $L$ is too big. Given a maximal Abelian Poisson subalgebra of $C^{\infty}(N)$, ie a maximal set of Poisson commuting observables, the corresponding prequantum operators commute and it is reasonable to assume that there will be a quantization of these observables on which they will act as multiplication operators. To achieve this, one considers a Hilbert space consisting of the sections of the prequantum line bundle which are covariantly 
constant along the Hamiltonian vector fields of those obsearvables. The generalization of this idea leads to the definition of polarization: an involutive Lagrangian distribution in $TN\otimes \mathbb{C}.$
The Hilbert space of quantum states will then be given by the sections of $L$ that are covariantly constant along the polarization, $P$. These are the polarized sections. A polarization $P$ is said to be complex if $P\cap \overline{P}=\{0\}$. Complex polarizations are closely related to Kähler structures and because of this connection they will be extremely useful to work with. Indeed, choosing a complex polarization on a symplectic manifold is equivalent to defining a compatible complex structure on the manifold (see \cite{BrianHall, Woodhouse}).

\begin{definition}Let $J$ be a compatible complex structure on $(N,\omega)$ so that $(N,\omega, J)$ is a K\"ahler manifold with Riemannian metric $\gamma (\cdot, \cdot) = \omega (\cdot, J\cdot).$ The corresponding K\"ahler polarization $P$ is defined by the complex polarization $P = T^{(1,0)}N$ so that, pointwise, $P$ is given by the  $(+i)$-eigenspace for $J$.
\end{definition}


{}

K\"ahler polarizations are very relevant for the study of the quantum Hall effect.
In order to obtain the correct physical results in some examples (see for example \cite{BrianHall}) and, in general, to get a better behaved quantization procedure, one also needs to include the so-called half-form correction. 
We note that, in the case of the two-dimensional torus and, more generally, arbitrary Riemann surfaces, half-forms are closely related to spin structures.

\begin{definition}
Let $P$ be a Kähler polarization on a symplectic manifold $(N,\omega)$. The canonical line bundle $K_{P}$ of P is the complex line bundle whose sections are the differential n-forms $\alpha\in \Omega^{n}(N;\mathbb{C})$ which satisfy:
\begin{equation}
    \iota_{X}\alpha = 0
    \label{halfformcondition}
\end{equation}
for all $X\in \Gamma(\overline{P})$.
\end{definition}

Considering the complex structure $J_{P}$ determined locally by holomorphic coordinates ($z_{1},...,z_{n}$) we see that the condition \eqref{halfformcondition} implies that the smooth sections of $K_{P}$ are precisely the $(n,0)$-forms on $N$.

\begin{definition}
A square root of $K_{P}$ is a complex line bundle $\delta_{P}$ over $N$ along with an isomorphism $\phi: \delta_{P}\otimes \delta_{P}\rightarrow K_{P}$.
\end{definition}

Thus, if $s_{1},s_{2}$ are sections of $\delta_{P}$ then $s_{1}\otimes s_{2}$ is a section of $K_{P}$. We will be working with  complex polarizations so $\delta_{P}$ is naturally a complex line bundle. The Lie derivative along 
$\overline{P}$ defines a partial connection on $K_{P}$ which descends to $\delta_{P}$; this allows for the definition of polarized sections of $\delta_{P}$. 
Assume now that we have chosen a square root of $K_{P}$, $\delta_{P}$, and that we also have a prequantum line bundle $L$ over $N$. We can then form the tensor product of these two line bundles $L\otimes\delta_{P}$. The sections of this line bundle can be decomposed in the following way: $\tilde{s}\in\Gamma(L\otimes\delta_{P})$ can be viewed, locally, as $\tilde{s}=s\otimes\nu$ with $s\in\Gamma(L)$ and $\nu\in\Gamma(\delta_{P})$.\par
We are now able to define a partial connection on the tensor bundle $L\otimes \delta_{P}$:

\begin{proposition}[(see \cite{BrianHall}]
Let P be a polarization on a symplectic manifold $(N,\omega)$ with prequantization $(L,h(.,.),\nabla)$. Let $\delta_{P}$ be a square root of $K_{P}$. The tensor product partial connection $\nabla: \Gamma(\overline{P})\times \Gamma(L\otimes\delta_{P})\rightarrow \Gamma(L\otimes\delta_{P})$ can be written as:
\begin{equation*}
    \nabla_{X}\tilde{s} = (\nabla_{X}s)\otimes \nu + s\otimes(\nabla_{X}\nu)
\end{equation*}
and it does not depend on the choice of $s,\nu$ and thus it is defined globally.
\end{proposition}

With the definition of a partial connection on the tensor bundle we can define polarized sections in this bundle as being the sections $s\in\Gamma(L\otimes\delta_{P})$ such that $\nabla_{X}s=0$, for each $X\in\Gamma(\overline{P})$.\par
The last element needed in order to define the Hilbert space for geometric quantization with half-form correction is to define an Hermitian structure on the bundle $\delta_{P}$, which we will denote by $\langle\cdot,\cdot\rangle$. It turns out that this Hermitian structure is naturally defined---because of the existence of the Liouville volume form---and so we also have a natural Hermitian structure on $L\otimes \delta_{P}$ by combining the two Hermitian structures.
We are now able to define the quantum Hilbert space with the half-form correction:

\begin{definition}
Let $P$ be a complex polarization on a quantizable symplectic manifold $(N,\omega)$ and let $(L,h(\cdot,\cdot),\nabla)$ be a prequantization of this manifold, where $h$ is an Hermitian structure on $L$ compatible with the connnection $\nabla$. Let $\delta_{P}$ be a square root of the canonical bundle $K_{P}$. Then, the half-form Hilbert space for a complex polarization $P$ on $N$ is:
\begin{equation*}
    \mathcal{H}_{P} := \overline{\{\tilde{s}\in\Gamma(L\otimes\delta_{P}):||\tilde{s}||^2<\infty, \nabla_{X}\tilde{s}=0,\forall X\in\Gamma(\overline{P})\}}^{||.||_{\mathrm{hf}}}
\end{equation*}

with 
$\langle \tilde{s}_{1},\tilde{s}_{2}\rangle_{\mathrm{hf}} := \int_{N}h(s_{1},s_{2})\langle \nu_{1},\nu_{2}\rangle$ and $||\tilde{s}||^{2}_{\mathrm{hf}} := \int_{N}h(\tilde{s},\tilde{s})|\nu|^2$, for $\tilde{s}_{1},\tilde{s}_{2},\tilde{s}\in\Gamma(L\otimes\delta)$.\par
$\mathcal{H}_{P}$ will be the space of square integrable polarized sections of $L\otimes\delta_{P}$.
\end{definition}

We will define a function $f\in C^{\infty}(N)$ to be quantizable if $X_{f}$ preserves $\overline{P}$. 

\begin{proposition}
Let P be a complex polarization on a quantizable symplectic manifold $(N,\omega)$ and let $(L,h(\cdot,\cdot),\nabla)$ be a prequantization of M. Let $\delta_{P}$ be a square root of the canonical bundle $K_{P}$ and $f\in C^{\infty}(N)$ be such that $X_{f}$ preserves $\overline{P}$. The prequantization of f on the half-form Hilbert space as follows:
\begin{equation}
Q_{\mathrm{pre}}(f)\tilde{s} := (Q_{\mathrm{pre}}(f)s)\otimes \nu + s\otimes i\hbar\mathcal{L}_{X_{f}}\nu,
\label{eq:halfformcorrectionprequantumoperator}
\end{equation}
with $\tilde{s}$ being locally written as $\tilde{s} = s\otimes \nu$, $s\in\Gamma(L)$ and $\nu\in\Gamma(\delta_{P})$, such that $Q_{\mathrm{pre}}$ does not depend on the choice of $s$ and $\nu$. In the above formula
\begin{equation}
Q_{\mathrm{pre}}(f)s=f\cdot s+i\hbar \nabla_{X_f}s  
\end{equation}
is the standard prequantization of $f$ in the absence of the half-form correction.
\label{correctionprequantum}
\end{proposition}

We then have that if $f\in C^{\infty}(N)$ is quantizable with half-form correction we define its quantization as:

\begin{equation}
Q(f) := Q_{\mathrm{pre}}(f)   
\end{equation}

These operators verify the commutation relation desired, $[Q(f),Q(g)]/(i\hbar)=Q(\{f,g\})$, on the space of polarized sections of $L\otimes\delta_{P}$. In addition, if $f$ is real valued and $X_{f}$ preserves $\overline{P}$, then $Q(f)$ will be symmetric. 

For the observables $g\in C^{\infty}(N)$ which are not polarization preserving we can, sometimes, still define a quantization operator by employing the following solution: if $g\in C^{\infty}(N)$ does not preserve the polarization but $g=g(f_{1},...,f_{n})$ with $n\in\mathbb{N}$ and all of $f_{i}$ preserve the polarization, then we can take the following operator to be the quantization of $g$:

\begin{equation}
    Q(g) := g(Q_{\mathrm{pre}}(f_{1}),...,Q_{\mathrm{pre}}(f_{n})) 
    \label{multiplicativityquantization}
\end{equation}

\subsection{Generalized Coherent State Transform}
\label{subs:gcst}

In Section \ref{sub:Hamiltonian_Flows_Imaginary_Time}, we described how, given an Hamiltonian function $H$, 
its Hamiltonian flow analytically continued to imaginary time can be used to deform K\"ahler structures, thus defining a continuous family of K\"ahler polarizations. 
In fact, this deformation can be lifted to the Hilbert spaces of half-form corrected polarized quantum states obtained via geometric quantization. Let the initial polarization of the system be $P_{0}$ evolve in complex time to another polarization $P_{\tau}$, with the two K\"ahler structures being related by the map in \eqref{eq:biholomorphism_complexstructure}; several examples show that it is natural to consider the  map:

\begin{equation}
    U_{\tau}: \mathcal{H}_{P_0}\rightarrow\mathcal{H}_{P_{\tau}}.
\end{equation}
where the map $U_{\tau}$ is the so-called generalized coherent state transform (gCST). For purely 
imaginary time $\tau = is$, 

\begin{equation}\label{timescst}
    U_{s} = \left(e^{-\frac{i}{\hbar}\tau Q_{\mathrm{pre}}(H)}\circ e^{\frac{i}{\hbar}\tau Q(H)}\right)\bigg|_{\tau=is}.
\end{equation}
(See \cite{matos:mera:mourao:mourao:nunes:2023} for the application to the fractional quantum Hall effect on the 2-sphere.)

A bit of the motivation for the form of the operator presented above will now be explained. The part of the operator that consists of the prequantum operator of $H$, $e^{-\frac{i}{\hbar}\tau Q_{\mathrm{pre}}(H)}$, is natural since it is a straightforward generalization of time evolution in quantum mechanics or, indeed, of the evolution of polarized sections along a real-time Hamiltonian flow. This map, as we shall see, when acting on the wave functions of the quantum Hall effect, does not preserve the Hilbert space of polarized sections but rather maps $J_{0}-$polarized states to $J_{\tau}-$polarized states. Now the other term, $e^{\frac{i}{\hbar}\tau Q(H)}$, comes as a correction. Indeed, for complex time $\tau$, the evolution by the prequantization operator of the Hamiltonian is non-unitary. Then, composition with the operator  $e^{\frac{i}{\hbar}\tau Q(H)}$, which preserves the Hilbert space, makes the whole operator $U_{s}$ unitary in several known examples or, at least, asymptotically unitary as $s\to +\infty$ in other examples.
 When the gCST is not unitary, the lack of unitarity can be seen as a measure of the non-equivalence of the quantum theories defined by the different choices of polarization. It will be with this operator that we will act on the quantum states of the quantum Hall effect.

\subsection{Theta Functions}
\label{sub:theta_functions}

In this section, we review the definition of theta functions and some of their properties, since these play an essential role in the description of the quantum Hall effect on the torus. For a detailed account of the theory of theta functions on Abelian varieties see, for example, 
\cite{kempf}.
The theta function is an entire function on $\mathbb{C}$ that is defined (for one variable) by:

\begin{equation}
\theta(z,\tau) = \sum_{n\in\mathbb{Z}} e^{\pi in^2\tau+2\pi inz},
\label{thetafunctiondef}
\end{equation}
with the parameter $\tau$ taking values in the upper-half plane, i.e. $\mathrm{Im}(\tau)>0$. The series  converges absolutely and uniformly on compact sets in the complex plane. These functions obey a quasi-periodicity condition with respect to the lattice $\Lambda=\mathbb{Z}\oplus \tau\mathbb{Z}$. Let us consider an element of the lattice above $\gamma= a + b\tau\in\Lambda$ with $a,b\in\mathbb{Z}$:
\begin{equation}
\theta(z+\gamma) = e^{-i\pi \tau b^2-2\pi ibz}\theta(z).
\label{quasiperiodicitythetafunction}
\end{equation}
This means that the theta function is an holomorphic section of an holomorphic line bundle over the complex torus (elliptic curve) $E_{\tau}=\mathbb{C}/\Lambda$, the theta line bundle. The complex torus $E_{\tau}$ is a Kähler manifold with a complex structure naturally induced by the standard complex structure on $\mathbb{C}$ and symplectic form given by 
$$\omega = 2\pi k dx \wedge dy, 
$$
where the integer $\, k$ is the level ($k=1$ for (\ref{thetafunctiondef})). The parameter $\tau$ determines the complex structure of the torus up to the modular group action: 
Two elliptic curves, $E_\tau, E_{\tau'}$ 
 are biholomorphic if and only if
the two moduli parameters $\tau$ and $\tau'$ are in the same orbit of the  modular group $\mathrm{SL}(2,\mathbb{Z})$. 

One can also define theta functions with characteristic, which give holomorphic sections of line bundles obtained from the theta bundle by a translation:

\begin{equation}
\vartheta\begin{bmatrix}
a\\
b
\end{bmatrix}(z,\tau) = \sum_{n\in\mathbb{Z}} e^{i\pi\tau(n+a)^2+2\pi i(n+a)(z+b)} = e^{i\pi a^2\tau+2i\pi a(z+b)}\theta(z+a\tau+b,\tau).
\label{eq:thetafunctionchardef}
\end{equation}

Comparing this with the definition of the theta function \eqref{thetafunctiondef} it is clear that the theta function with characteristic 0 reduces to the theta function:

\begin{equation*}
\vartheta\begin{bmatrix}
0 \\
0 
\end{bmatrix}(z,\tau) = \theta(z,\tau).
\end{equation*}

It is also convenient to know their quasi-periodicity properties, so let $\gamma = c+d\tau\in \Lambda$. We have that:

\begin{equation*}
\vartheta\begin{bmatrix}
a\\
b
\end{bmatrix}(z+\gamma,\tau) = e^{-i\pi\tau d^2-2i\pi dz}e^{2i\pi(ac-db)}\vartheta\begin{bmatrix}
a\\
b
\end{bmatrix}(z,\tau).
\end{equation*}

Comparing the multipliers of the theta functions with the ones just obtained for the theta functions with characteristics we have that they differ by a flat multiplier:
\begin{equation*}
    \chi_{\gamma} = e^{2i\pi(ac-db)},
\end{equation*}
with $\gamma = c+d\tau\in\Lambda$. These unitary multipliers $(\chi_{\gamma})_{\gamma\in\Lambda}$ determine a unitary character of the lattice:
\begin{align*}
    \chi :& \Lambda\rightarrow U(1)\\
    &\gamma\mapsto \chi_{\gamma}.
\end{align*}

Moreover, meromorphic functions on $E_\tau$ are then described by quotients of theta functions. This is the content of Abel's theorem. 

To sum up this section, it is also useful to introduce, in light of the present work, a theta function with specific characteristic:

\begin{equation}
\theta_{11}(z,\tau) = \vartheta\begin{bmatrix}
\frac{1}{2} \\
\frac{1}{2}
\end{bmatrix}(z,\tau) = e^{\frac{i\pi\tau}{4}+2i\pi nz}\theta\left(z+\frac{1}{2}(1+\tau),\tau\right).
\end{equation}

This theta function has zeros at $\Lambda$ and their quasi periodicity conditions for $\gamma = a+b\tau$ are given by the following equation:

\begin{equation}
\theta_{11}(z+\gamma,\tau) = (-1)^{a+b}e^{-i\pi\tau b^2-2\pi ibz}\theta_{11}(z,\tau). 
\label{eq:theta11quasiperiodicity}
\end{equation}

\section{The Lowest Landau Level for flat geometry and non-periodic quadratic Hamiltonian}
\label{sec:flat_geometry}
We consider electrons living on a Kähler surface $(X,iF, J,\gamma)$ and subject to a
uniform external magnetic field. The most natural way of expressing the magnetic field in geometric terms is to consider a 2-form $F\in\Omega^{2}(X;\mathbb{C})$, which we will assume to be closed and non-degenerate. We consider a Hermitian holomorphic line bundle $L\rightarrow X$ over the Riemann surface, with  Chern connection 
$\nabla$ and curvature given by $F$. The fact that the magnetic field is uniform corresponds to the expression  $iF = \omega = B\sqrt{\det(g)}dx\wedge dy$, in local coordinates $(x,y)$, with $B = \frac{ 2\pi\deg(L)}{\mathrm{area}(X)}$ and $\omega$ is the area form on the surface. Throughout the remainder of this work, we set Planck's constant in the geometric quantization formalism to $\hbar = 1$.

The single-particle Hilbert space of the quantum theory is described by the square integrable sections of the line bundle $L$. The single-particle Hamiltonian is described
by the Bochner Laplacian (also known as the magnetic Laplacian in this context), whose groundstate subspace, known in the physics literature
as the lowest Landau level (LLL), is the space of holomorphic sections of $L$,  $H^{0}(X,L)$.  Geometric quantization allows us, therefore, to identify the lowest Landau level with the Hilbert
space of quantum states for the Kähler quantization of $X$.

In this section, we describe the LLLs on the torus with flat geometry deformed by Hamiltonian flows in imaginary time. In Section \ref{section:oneparticlestates} we describe the process in detail. We determine explicit the single-particle holomorphic states and apply the gCST to them. We use these results to study the particle density profile. 

We then apply an analogous procedure for the multiparticle states. We start by studying, in Section \ref{section:integerquantumhalleffect}, the integer quantum Hall effect, in which the groundstate subspace is fully filled. 
We present some of the mathematical notions used to represent many-particle states, as well as how they evolve by imaginary time Hamiltonian flow. We use these results to study the particle density profile and compare the results with the ones in literature.

In Section \ref{section:fractionalquantumhalleffect}, we describe the FQHE, the wavefunctions here are the Laughlin states. These are states of $N_{e}$ electrons and filling fraction of $\nu=\frac{N_{e}}{N_{\phi}}=\frac{1}{k}$ with $k$ an odd number, where $N_{\phi} = \phi/\phi_{0} = B\cdot \mathrm{area}(X)/(2\pi)$ is the number of flux quanta. Here, we understand how these states evolve by the imaginary time Hamiltonian flow. We use these results to study the particle density profile. At the end of the Section, we comment on the deformation of the Laughlin state, induced by the deformation of geometry along the the Hamiltonian flow in imaginary time flow that we consider, and the Tao-Thouless state (for filling fraction $\nu= 1/3$).

\subsection{One-Particle States}
\label{section:oneparticlestates}

We consider an electron subject to a uniform magnetic field on an elliptic curve $X=E_{\tau}=\mathbb{C}/\Lambda$, with $\Lambda = \mathbb{Z}\oplus \tau\mathbb{Z}$ and $\tau=\tau_{1}+i\tau_{2},\, \tau_2>0$. 

Following the prescription of geometric quantization, the symplectic form on $E_{\tau}$ will be integral, as explained in Section \ref{sub:geometric_quantization}, and it will also contain the information about the external magnetic field. The relation between the magnetic field $B$ and the total number of flux quanta is $N_{\phi}=\phi/\phi_{0} = B\cdot \mathrm{area}(E_\tau)/(2\pi)$. The symplectic form on $E_{\tau}$ is then, in local flat coordinates coordinates $(x,y)$ induced from the quotient $\C/\Lambda$, $$\omega = 2\pi N_{\phi}dx\wedge dy.$$ The corresponding K\"ahler metric on $E_\tau$ is then $$g = \frac{2\pi N_{\phi}}{\tau_{2}}|dz|^2,$$ with $z=x+\tau y$ the holomorphic coordinate in the universal cover $\mathbb{C}$. 

We can now construct the line bundle, $L^{N_{\phi}}$, whose space of square-integrable sections gives the single-particle Hilbert space. $L^{N_{\phi}}$ can be defined as a quotient of the trivial complex line bundle over $\C$, so that its sections can be viewed as quasi-periodic functions, with respect to $\Lambda,$ on $\C$.  
The choice of trivializing section over $\C$ corresponds to a choice of gauge.
  We will work in a unitary gauge, also called the Landau gauge, commonly used the QHE literature, so that the inner product has a simpler form, which will be useful for normalizing the wavefunctions. Consider the system of unitary multipliers---determining the $\Lambda$-action defining $L^{N_{\phi}}=\left(\mathbb{C}\times \mathbb{C}\right)/\Lambda$---given by
\begin{equation}
\widetilde{e}^{N_{\phi}}_{\gamma}:\mathbb{C}\rightarrow \mathrm{U}(1)\subset  \mathbb{C}^{*},
\end{equation}
for $\gamma=a+b\tau\in\Lambda$ and
\begin{equation}
   \widetilde{e}^{N_{\phi}}_{\gamma = a+b\tau}(z) = e^{-2\pi i N_{\phi}bx}.
   \label{multipliersunitary}
\end{equation}
The tilde notation $\widetilde{e}^{N_{\phi}}_\gamma$ is employed to emphasize the unitary nature of these multipliers, distinguishing them from multipliers $e^{N_{\phi}}_\gamma$ in a holomorphic gauge to appear later. It is simple to check that the  $1$-form given by $$A(z) = 2\pi iN_{\phi}y dx.$$ is consistent with the above system of multipliers, i.e.,
\begin{equation}
  A(z+\gamma) = A(z) -  d\log \widetilde{e}^{N_{\phi}}_{\gamma}(z),
\end{equation}
and, thus, determines a connection on $L^{N_{\phi}}$. 
The Chern connection is given by
\begin{equation}
    \nabla_{Y}s = Y\cdot s + A(Y)\cdot s,
\end{equation}
for a vector field $Y$ and section $s$. 

We will now consider the K\"ahler polarization of $E_\tau$, in the formalism of geometric quantization, to obtain the Hilbert space of (polarized) quantum states.  The Kähler polarization which is given by the distribution
\begin{equation}
  P = \Big\langle \frac{\partial}{\partial z}\Big\rangle. 
\end{equation}
In order to determine which of the sections $s$ of the line bundle are polarized, and hence will belong to the quantum Hilbert space of the system, we need to solve the equation for holomorphic sections
\begin{equation}
\nabla_{\frac{\partial}{\partial\bar{z}}}s = 0 \Rightarrow s = s(z).
\end{equation}
Using the connection 1-form in the unitary gauge $A(z) = 2\pi iN_{\phi}y dx$ we obtain that:
\begin{align*}
&\nabla_{\frac{\partial}{\partial\bar{z}}}s=0 \Leftrightarrow s(z) = e^{i\pi N_{\phi}\tau y^2}f(z),
\end{align*}
where $f(z)$ is a holomorphic function. It turns out that these holomorphic functions are the theta functions with characteristics in Eq.\eqref{eq:thetafunctionchardef}, $\theta_{l/N_{\phi}}(N_{\phi}z,N_{\phi}\tau):=\vartheta\begin{bmatrix}
\tfrac{l}{N_{\phi}}\\
0
\end{bmatrix}(N_{\phi}z,N_{\phi}\tau)$, with $l=0,...,N_{\phi}-1$. To prove this, one just needs to check that the section $s$ with $$f(z) = \theta_{l/N_{\phi}}(N_{\phi}z,N_{\phi}\tau)$$ transforms in the correct way with the unitary multipliers given by \eqref{multipliersunitary}, i.e. $$s(z+\gamma)=\widetilde{e}^{N_{\phi}}_{\gamma}(x)\cdot s(z).$$  
This defines 

The holomorphic sections in the Landau gauge are then given by
\begin{equation}
s\in\Gamma(E_{\tau},L^{N_{\phi}})\Rightarrow \psi^{(l)}_{\mathrm{LLL}}(z,\tau):= s^{(l)}(z) = e^{i\pi N_{\phi}\tau y^2}\theta_{l/N_{\phi}}(N_{\phi}z,N_{\phi}\tau).
\label{LLLunitary}
\end{equation}
These polarized sections form, up to an overall normalization independent of $\ell$, a unitary basis of $H^0(E_\tau, L^{N_\phi})$, and they are the single-particle wavefunctions of the lowest Landau level used in this work.

\subsection{Evolution of one-particle states under the gCST}
\label{sub:evolutionsingleparticle}
In this Section, our aim is to describe how the polarized sections of $L^{N_{\phi}}$ evolve under the generalized coherent state transform (gCST) corresponding to an Hamiltonian flow in imaginary time generated by the Hamiltonian function $$H(x,y)= \frac{y^2}{2}.$$ Note that $H$ is actually well-defined only 
on the universal covering $\C$ of $X=E_\tau$. As we will see, its Hamiltonian flow induces a change in the biholomorphism class of $X$ given by a change in the modular parameter $\tau$ along the flow.

As explained in Section \ref{sub:geometric_quantization}, the evolution of the lowest Landau level (LLL) should include the half-form correction so that we now consider the line bundle $L^{N_{\phi}}\otimes K^{1/2}$, where $K$ is the canonical line bundle of $E_{\tau}$ with respect to the chosen Kähler polarization. Since $dz$ is a global trivialization of $K$, we can choose the global trivialization of $K^{1/2}$ to be $\sqrt{dz}$.

The sections of the line bundle $L^{N_{\phi}}\otimes K^{1/2}$ take the form $s\otimes \nu$, where $s \in \Gamma(L^{N_{\phi}})$ and $\nu \in \Gamma(K^{1/2})$. The prequantum operators,  from Eq.~\eqref{eq:halfformcorrectionprequantumoperator}, give that the gCST takes the form
\begin{equation}
   U_{s} = U_{1,s} \otimes U_{2,s}=\left(e^{-itQ_{\mathrm{pre}}(H)}e^{\frac{it}{2}(Q_{\mathrm{pre}}(y))^2} \otimes 
   e^{-it\left(i\mathcal{L}_{X_{H}}\right)}  \right)\bigg|_{t=is},
\end{equation}
where the second factor $U_{2,s}$ in the tensor product corresponds to the evolution of the half-form. In the equation above, we observe that $X_y$ preserves $dz$, so no additional contribution to $U_{2,s}$ arises.

Considering the connection 1-form $A = 2\pi i N_{\phi} y dx$ on the line bundle $L^{N_{\phi}}$, the prequantum operator associated with $H$ is
\begin{equation}
   Q_{\mathrm{pre}}(H) = H + i\nabla_{X_{H}} = H +iX_{H}-2\pi N_{\phi}ydx(X_{H}) = H +iX_{H}- y^2 = -H + iX_{H},
\end{equation}
with $X_{H} = \frac{y}{2\pi N_{\phi}} \frac{\partial}{\partial x}$ as the Hamiltonian vector field. Similarly, the prequantum operator associated with $y$ is:

\begin{equation}
    Q_{\mathrm{pre}}(y) = y+i\nabla_{X_{y}} = y + iX_{y}-2\pi N_{\phi}ydx(X_{y}) = y + iX_{y} - y = iX_{y},
\end{equation}
where $X_{y} = \frac{1}{2\pi N_{\phi}} \frac{\partial}{\partial x}$. Finally, the quantum operator of $H$ will be given by:

\begin{equation}
    Q(H) = \frac{1}{2}(Q_{\mathrm{pre}}(y))^2 = -\frac{1}{2}\frac{1}{(2\pi)^2}\frac{1}{(N_{\phi})^2}\frac{\partial^{2}}{\partial x^2}.
\end{equation}
The first factor of the gCST is then
\begin{equation}
   U_{1,s}=\left(e^{it\frac{y^2}{2}}e^{t\frac{y}{2\pi N_{\phi}}\frac{\partial}{\partial x}}e^{\frac{-it}{2N^2_{\phi}}\frac{1}{(2\pi)^2}\frac{\partial^2}{\partial x^2}}\right)\bigg|_{t=is}.
   \label{eq:GCST1}
\end{equation}
Acting with this operator on the wavefunctions of the LLL in Eq.~\eqref{LLLunitary}, we obtain the result:
\begin{equation}
 U_{1,s}
\left(\theta_{l/N_{\phi}}(N_{\phi}z,N_{\phi}\tau)e^{i\pi N_{\phi}\tau y^2}\right) = \theta_{l/N_{\phi}}(N_{\phi}z_{s},N_{\phi}\tau_{s})e^{i\pi\tau_s N_{\phi}y^2},
\label{eq:evolutionLLL}
\end{equation}
with $\tau_{s} = \tau + \frac{is}{2\pi N_{\phi}}$ and $z_{s} = z + is \frac{y}{2\pi N_{\phi}} = x + \tau_{s} y$. 

A few remarks on the action of the first factor of the gCST on the LLL wavefunctions are necessary. The prequantum operator in Eq.~\eqref{eq:GCST1} is responsible for changing the complex coordinate $z$ to $z_s$, while the quantum operator changes the modular parameter $\tau$ to $\tau_s$. Furthermore, this operator transforms the section of the line bundle $L^{N_{\phi}} \to E_{\tau}$, corresponding to $\theta_{l/N_{\phi}}(N_{\phi}z, N_{\phi}\tau) e^{i\pi N_{\phi} \tau y^2}$, into a section of the line bundle $L^{N_{\phi}}_{s} \to E_{\tau_{s}}$, given by $\theta_{l/N_{\phi}}(N_{\phi}z_{s}, N_{\phi} \tau_{s}) e^{i\pi N_{\phi} \tau_{s} y^2}$. The multipliers of this new line bundle, in the holomorphic gauge, are
\begin{equation}
    e^{N_{\phi}}_{\gamma}(z_{s}) = e^{-i\pi N_{\phi}\tau_{s}b^2-2\pi ibN_{\phi}z_{s}}.
    \label{multipliersevolved}
\end{equation}
The factor $e^{-s\frac{y^2}{2}}$ is necessary to correct the gauge factor of the unitary section, which is $e^{i\pi N_{\phi} \tau y^2}$. 

Note that in the formalism for imaginary time Hamiltonian flows that we reviewed briefly in Section \ref{sub:Hamiltonian_Flows_Imaginary_Time}, one considers a smooth Hamiltonian on $X$ whose flow in imaginary time gives rise to different, but equivalent or biholomorphic, complex structures. In the present situation, we have a non-trivial generalization of this formalism since the Hamiltonian function $H$ is smooth on the universal covering $\C$ of $X=E_\tau$. As described above, however, this Hamiltonian flow in the universal cover induces a non-trivial deformation of the modular parameter, from $\tau$ to $\tau_s$, thus changing the biholomorphism class of $X$. Of course, the generalization of this phenomenon to other families of K\"ahler manifolds would be very interesting to explore.

We now compute the $L^2$-squared norm of the evolved section $s_{s}$. The normalization of the wavefunction is essential because we must work with normalized wavefunctions when considering particle density later. In the unitary gauge, the normalization is found to be (using standard Gaussian integral techniques)
\begin{align}
||s_{s}||^2_{L^2} &= \int_{0}^{1}\int_{0}^{1}dxdy \Big|U_{1,s}(\theta_{l/N_{\phi}}(N_{\phi}z,N_{\phi}\tau)e^{i\pi N_{\phi}\tau y^2})\Big|^2 \nonumber\\
&= \frac{1}{\sqrt{2N_{\phi}\tau_{2_s}}}.
\label{eq:normalizationunitary}
\end{align}
From this computation, we see that the normalization of the sections depends only on the imaginary part of the modular parameter, $\tau_{2}$. 

Now that we understand how the first factor of the gCST, $U_{1,s}$, acts on the LLL wavefunctions, let us study the action of the second factor $U_{2,s}$ on the half-forms. Note that $\mathcal{L}_{X_{H}}$ acts on half-forms as follows:
\begin{align*}
    \mathcal{L}_{X_{H}}(\sqrt{dz}\sqrt{dz}) &=  \mathcal{L}_{X_{H}}(\sqrt{dz})\sqrt{dz} + \sqrt{dz}\mathcal{L}_{X_{H}}(\sqrt{dz})\\
    &= 2\mathcal{L}_{X_{H}}(\sqrt{dz})\sqrt{dz},
\end{align*}
where we used the commutativity of the tensor product in one-dimensional vector spaces. Thus:
\begin{equation}
\mathcal{L}_{X_{H}}(dz) = d(\iota_{X_{H}}dz) = \frac{dy}{2\pi N_{\phi}} \Rightarrow e^{isL_{X_H}}dz = d\left(z+\frac{isy}{2\pi N_{\phi}}\right)= dz_{s}
\label{evolutionhalfform}
\end{equation}
so that, for $t=is,$
\begin{equation}
 U_{2,s} \sqrt{dz} = \sqrt{dz_{s}}.
 \label{halfformevolution}
\end{equation}

Combining the result of the LLL evolution with the first factor of the gCST Eq.~\eqref{eq:evolutionLLL} and the half-form evolution under the second factor Eq.\eqref{halfformevolution}, we obtain, 
\begin{equation}
U_{t}\left(\theta_{l/N_{\phi}}(N_{\phi}z,N_{\phi}\tau)e^{i\pi N_{\phi}\tau y^2}\otimes dz^{1/2}\right)\bigg|_{t=is}= \theta_{l/N_{\phi}}(N_{\phi}z_{s},N_{\phi}\tau_{s})e^{i\pi\tau_{s} N_{\phi}y^2}\otimes dz_{s}^{1/2}.
\end{equation}
Since $L^{N_{\phi}}$ is an Hermitian line bundle, the square root of the canonical line bundle $K^{1/2}$ inherits an Hermitian structure, with the Hermitian metric denoted by $h$:

\begin{equation}
h(\sqrt{dz}) = ||dz^{1/2}||^{2} = \sqrt{\frac{\tau_{2}}{2\pi N_{\phi}}}.  
\end{equation}

We can now compute the norm square of $s\otimes dz^{1/2}$, to obtain
\begin{align}
||s_{s}\otimes dz_{s}^{1/2}||^{2}_{L^2} &= \int_{E_{\tau}}dx\wedge dy \; h(s_{s}\otimes dz_{s}^{1/2})\nonumber\\
&= \frac{1}{2N_{\phi}\sqrt{\pi}}.
\label{eq:normtotalsection}
\end{align}

From Eq.\eqref{eq:normtotalsection}, we conclude that the norm is independent of $s$, which proves that, with the half-form correction, the gCST is a unitary isomorphism of the quantum Hilbert spaces associated with the polarizations $\left\langle \frac{\partial}{\partial z} \right\rangle$ and $\left\langle \frac{\partial}{\partial z_{s}} \right\rangle$. 

The limit $s\rightarrow \infty$ is an interesting case to study in the evolution described earlier. Under this extreme deformation of the torus geometry, Dirac delta distributions naturally arise, as there is a convergence of the sections of the bundle $L^{N_{\phi}}\otimes K^{1/2}$, in a distributional sense, to Dirac delta functions centered at the points $n+l/N_{\phi}$, with $n\in\mathbb{Z}$. These points correspond to Bohr-Sommerfeld cycles in the geometric quantization of $X$ in a real polarization to which the K\"ahler polarizations converge as $s\to\infty$.
(See \cite{BaierMouraoNunes} for a general discussion of the quantization of complex tori along families of flat K\"ahler polarizations and of the corresponding degeneration to real polarizations with the consequent appearance of distributional wave functions.)

 Considering the case where the initial modular parameter takes the (singular) value $\tau = 0$, in the limit of $s\rightarrow \infty$, the evolution of the wavefunctions with the half-form correction becomes
\begin{align*}
\lim_{s\rightarrow \infty} U_{s}\left(\theta_{l/N_{\phi}}(N_{\phi}z,0)\otimes dz^{1/2}\right) &= \sum_{n\in\mathbb{Z}}\frac{1}{\sqrt{N_{\phi}}}\delta\left(y+\left(n+\frac{l}{N_{\phi}}\right)\right)e^{2\pi ixN_{\phi}\left(n+\frac{l}{N_{\phi}}\right)}\otimes dy^{1/2},
\end{align*}
with $l=0,...,N_{\phi}-1$. These correspond to the (distributional) polarized sections for the real polarization $\Big\langle \frac{\partial}{\partial x}\Big\rangle$. The infinite imaginary time gCST, which we will describe below, up to a total phase factor, gives the discrete Fourier transform in $\mathbb{Z}_{N_{\phi}}$. This limit $s\to\infty$ of the one-particle states 
gives, see Sec.~\ref{subsec:EvolutionFQHE}, an adiabatic deformation between the Laughlin state and the Tao-Thouless state (for $\nu=\frac13$) for which the torus has the geometry of an infinitely thin cylinder.

Before analyzing the evolution of many-particle states, we compute the electronic density, which in first quantization takes the following form
\begin{equation}
    \rho(z) = \psi(z)\overline{\psi(z)},
    \label{densitysecondquantization}
\end{equation}
with $\psi$ the normalized wavefunction.
We are interested in how the particle density evolves as $s\rightarrow \infty$, a limit equivalent to $\tau_{2_s}\rightarrow \infty$. The effective density along the $y$-direction is obtained by integration over $x$
\begin{align}
\rho(y) &= \sqrt{2N_{\phi}\tau_{2_s}}\int_{0}^{1}dx\Big|\theta_{l/N_{\phi}}(N_{\phi}z_s,N_{\phi}\tau_s)e^{i\pi N_{\phi}\tau_s y^2}\Big|^2\nonumber\\
&= \sqrt{2N_{\phi}\tau_{2_s}}\sum_{n\in\mathbb{Z}}e^{-2\pi N_{\phi}\tau_{2_s}\left(y+\left(n+\frac{l}{N_{\phi}}\right)^2\right)}.
\label{eq:particledensitysingle}
\end{align}
In the limit $\tau_{2_s}\rightarrow \infty$, the integrated density becomes a Dirac delta distribution:

\begin{equation}
    \rho(y) = \sum_{n\in\mathbb{Z}} \delta\left(y+n+\frac{l}{N_{\phi}}\right)
    \label{integrateddensitysingleparticle}
\end{equation}
with $l=0,...,N_{\phi}-1$. Hence, in the limit $\tau_{2_s}\rightarrow \infty$ the density profile of the single-particle state approaches Dirac delta distributions centered at the points $n+l/N_{\phi}$, with $n\in\mathbb{Z}$. In Fig.\ref{fig:singleparticledensity}, we give a numerical illustration of this fact.

\begin{figure}[h]
    \centering
    \includegraphics{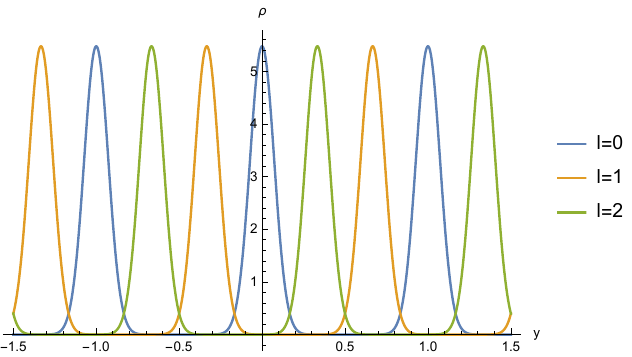}
    \caption{Density Profile for single particle states- $N_{\phi}=3$- $\tau = 5i$}
    \label{fig:singleparticledensity}
\end{figure}

Numerically, the limit $\tau_{2_s}\rightarrow \infty$ corresponds to values of $\tau_{2_s}$ that are not excessively large due to convergence issues with the theta functions. Thus, in the plot in Fig.\ref{fig:singleparticledensity}, we used $\tau_{2_s}=5i$. From the figure, it is clear that Gaussians are centered at the points $y=0,1,2,\dots$ for $l=0$, $y=-1/3,2/3,\dots$ for $l=1$ and $y=-2/3,1/3,\dots$ for $l=2$, which is consistent with the result for the integrated density of the single-particle states in Eq.~\eqref{integrateddensitysingleparticle}.

\subsection{Integer Quantum Hall Effect}
\label{section:integerquantumhalleffect}

In this section, we begin examining multiparticle states by focusing on the integer quantum Hall effect (IQHE). In these systems, the lowest Landau level (LLL) is fully occupied, corresponding to a filling fraction $\nu=1$, where the number of electrons equals the total flux quanta, $N_{e}=N_{\phi}\equiv N$. The fully filled LLL state corresponds to a unique state $\psi\in \Lambda^{N}H^{0}(E_{\tau},L^{N})$ (since $\Lambda^{N}H^{0}(E_{\tau},L^{N})$ is $1$-dimensional), where $L^{N}$ is the line bundle constructed in the previous section. Here, $\psi\in \Lambda^{N}H^{0}(E_{\tau},L^{N})$ can be understood as a fully antisymmetric section $\psi(z_{\sigma(1)},\dots,z_{\sigma(N)})=\mathrm{sgn}(\sigma)\psi(z_{1},\dots,z_{N})$, $\sigma\in S_N$ , of the external $N$th tensor product $L^{\boxtimes N}$ 
\begin{equation}
L^{\boxtimes N} = L\boxtimes \dots\boxtimes L = p_{1}^{*}L\otimes \dots \otimes p_{N}^{*}L \longrightarrow X\times\dots \times X = X^{N},
\label{externallinebundle}
\end{equation}
where $p_{i}$ is the $i$th canonical projection in the $N$-fold Cartesian power of $X$, $i=1,\dots, N$.

To construct the sections of $L^{\boxtimes N}$ we could use the basis previously found for $H^{0}(E_{\tau},L^{N})$, as $L^{\boxtimes N}$ is formed by the tensor product of $L$.  
In the unitary gauge basis for $H^{0}(E_{\tau}, L^{N})$, as in Eq.\eqref{LLLunitary}, we can represent a section of $L^{\boxtimes N}$ as:
\begin{equation*}
    \psi(z_{1},\dots,z_{N}) = \det\left[\theta_{i/N}(Nz_{j},N\tau)e^{i\pi N \tau y^2_j}\right]_{1\leq i,j \leq N}.
\end{equation*}
However, in order to address more general filling fractions, we adopt a different approach. Fixing $N-1$ variables (e.g., all except $z_{i}$), $\psi$ transforms as a section of $L^{N}$. Additionally, because $\psi$  represent a fermionic wavefunction, it is antisymmetric and has $N-1$ zeros at the positions of complementary variables. Thus, we can express $\psi$ as:
\begin{equation}
    \psi(z_{1},\dots,z_{N}) = f(z_{1},\dots,z_{N})\prod_{1\leq i,j\leq N}\theta_{11}(z_{i}-z_{j},\tau)e^{i\pi \tau N \sum_{j=1}^{N}y^2_j}.
    \label{preinteger}
\end{equation}
Given that $\theta_{11}(z_{i}-z_{j},\tau)\sim z_{i}-z_{j}$ for small $z_{i}-z_{j}$, it correctly reflects the $N-1$ zeros in the complementary positions. Furthermore, since $\theta_{11}(z_{i}-z_{j},\tau)$ is antisymmetric, $f(z_{1},\dots,z_{N})$ must be a symmetric to ensure that $\psi(z_{1},\dots,z_{N})$ is antisymmetric as required. To determine $f(z_{1},\dots,z_{N})$, we use the fact that it should transform such that fixing $N-1$ variables results in a section of $L^{N}$. For $\gamma = a+b\tau\in\Lambda$:
\begin{align*}
&\frac{\psi(z_{1}+\gamma,\dots,z_{N})}{\psi(z_{1},\dots,z_{N})} \nonumber\\
&= \frac{f(z_{1}+\gamma,\dots,z_{N})}{f(z_{1},\dots,z_{N})}(-1)^{(N-1)(a+b)}e^{-i\pi(N-1)\tau b^2}e^{-2i\pi b\sum_{j=2}^{N}(z_{1}-z_{j})}e^{2\pi i\tau N b y_1+i\pi \tau N b^2} \\
&= \widetilde{e}_{\gamma}^{N}(z_{1}) = e^{-2 i\pi N b x_1},  
\end{align*}
where we used Eq.~\eqref{eq:theta11quasiperiodicity} which relates $\theta_{11}(z_{i}-z_{j}+\gamma,\tau)$ and $\theta_{11}(z_{i}-z_{j},\tau)$. We conclude that
\begin{equation}
\frac{f(z_{1}+\gamma,\dots,z_{N})}{f(z_{1},\dots,z_{N})} = (-1)^{(N-1)(a+b)}e^{-i\pi\tau b^2}e^{-2\pi ib\sum_{j=1}^{N}z_{j}}.   
\label{eq:quasiperiodicinteger}
\end{equation}
The right-hand side defines a multiplier for a bundle over $E_{\tau}$, resembling that for $L^{N_{\phi}}$, but twisted by a flat bundle defined by the character $\chi_{\gamma} = (-1)^{(N-1)(a+b)}$. Defining the map $g: E_{\tau}^{\times N}\rightarrow E_{\tau}; ([z_{1}],...,[z_{N}])\mapsto \left[\sum_{j=1}^{N}z_{j}\right]:=[Z]$, we interpret $f$ as the pullback under $g$ of a section of $L$ over $E_{\tau}$ twisted by a flat line bundle. There exists a unique section, up to multiplication by a non-vanishing scalar, that satisfies the quasi-periodicity condition in Eq.\eqref{eq:quasiperiodicinteger}, given by
\begin{equation*}
    f(z_{1},\dots ,z_{N}) = \vartheta\begin{bmatrix}
\frac{N-1}{2}\\
\frac{N-1}{2}
\end{bmatrix}(Z,\tau). 
\end{equation*}
Thus, the wave function for the fully filled LLL of the IQHE is
\begin{equation}
    \Psi_{\mathrm{Integer}}(z_{1},...,z_{N}) = \vartheta\begin{bmatrix}
\frac{N-1}{2}\\
\frac{N-1}{2}
\end{bmatrix}(Z,\tau)\prod_{1\leq i,j\leq N}\theta_{11}(z_{i}-z_{j},\tau)e^{i\pi \tau N \sum_{j=1}^{N}y^2_j}. 
\label{eq:integerwavefunction}
\end{equation}

\subsection{Evolution of the fully filled LLL states under the   gCST}
\label{subsec:evolutionintegerquantumhalleffect}

 We now examine the evolution of IQHE wavefunctions under the gCST with the half-form correction, following from the evolution of the single-particle states described in Section \ref{sub:evolutionsingleparticle}. 
 
 We start by obtaining a trivializing section of the square root of the canonical line bundle $K^{1/2}$ associated with $L^{\boxtimes N}$. Since $L^{\boxtimes N}$ is the external $N$ fold tensor product of $L$, the trivializing sections of the square root of the canonical line bundle associated with $E_{\tau}^{N}$ can be directly obtained from those associated with $E_{\tau}$. In Section \ref{sub:evolutionsingleparticle}, we found $\sqrt{dz}$ is the trivializing section of $K^{1/2}$ for $E_{\tau}$ so that $\sqrt{dz_{1}\wedge...\wedge dz_{N}}$ is a trivializing section of $K^{1/2}$ for $E_{\tau}^{N}$. As with the single-particle case, the gCST consists of two factors: one that acts on the integer quantum Hall states in Eq.\eqref{eq:integerwavefunction} and another that acts on the sections of the square root of the canonical line bundle $K^{1/2}$. The gCST acts naturally on products of one-particle states, and thus its action on the half-form correction is straightforward:
\begin{equation}
    U_{2,s}(\sqrt{dz_{1}\wedge...\wedge dz_{N}}) = \sqrt{d{z_s}_{1}\wedge...\wedge d{z_s}_{N}} 
    \label{eq:evolutionhalfformIQHE},
\end{equation}
where the Lie derivative in $U_{2,s}$ acts on each coordinate independently.

To determine how  the gCST acts on the wavefunction of the integer quantum Hall states in Eq.~\eqref{eq:integerwavefunction}, we note that this wavefunction is not a product of single-particle states evaluated at each $z_{j}$, which complicates the evolution. The most intuitive approach is to first evolve $\prod_{1\leq i,j\leq N}\theta_{11}(z_{i}-z_{j},\tau)$ with a tensor product of gCST, acting on the coordinates $U_{ij}=z_{i}-z_{j}$. Each term in the product evolves as a one-particle function under the gCST, modifying the $z_{i}-z_{j}$ and $\tau$ coordinates as follows:
\begin{align}
 \label{alterationtaus}
 &\tau\rightarrow \tau_{s} = \tau + \frac{is}{2\pi N}\\
 & z_{i}-z_{j} \rightarrow (z_{i}-z_{j})_{s} = z_{i}-z_{j}+\frac{is}{2\pi N}(y_{i}-y_{j}). 
 \label{alterationzs}
\end{align}

For the center-of-mass pre-factor, we can treat $Z$ as the coordinate of a ``particle'' that evolves with the gCST similarly to a single-particle state. Thus, applying the gCST results in the following transformations:
\begin{align}
&\tau\rightarrow\tau_{s} = \tau + \frac{is}{2N\pi}\nonumber\\
&Z\rightarrow Z_{s} = Z+\frac{is}{2N\pi}Y,
\label{alterarionZ}
\end{align}
where $Y=\sum_{j=1}^{N}y_j$. This matches the gCST’s effect on the center-of-mass term in the IQHE wavefunction, altering $Z$ and to $\tau$ as described. The gCST has the following form:
\begin{equation*}
U^{\mathrm{CM}}_{1,s} = \left(e^{it\sum_{i=1}^{N}\frac{y_{i}^{2}}{2}}e^{\frac{t}{2\pi N}\sum_{i=1}^{N}y_{i}\frac{\partial}{\partial x_{i}}}e^{-\frac{it}{2N^{2}}\frac{1}{(2\pi)^2}\sum_{i=1}^{N}\frac{\partial^2}{\partial x_{i}^2}}\right)\bigg|_{t=is}.     
\end{equation*}
Thus,
\begin{align*}
U^{\mathrm{CM}}_{1,s}\left(\vartheta\begin{bmatrix}
\frac{N-1}{2}\\
\frac{N-1}{2}
\end{bmatrix}(Z,\tau)e^{i\pi \tau N\sum_{j=1}^{N}y^2_j}\right) 
&= \vartheta\begin{bmatrix}
\frac{N-1}{2}\\
\frac{N-1}{2}
\end{bmatrix}\left(Z+\frac{isY}{2\pi N},\tau+\frac{is}{2\pi N}\right)e^{i\pi \tau_{s} N\sum_{j=1}^{N}(y_{j,s})^2}.
\end{align*}
In summary, evolving the multiparticle wave function given by Eq.~\eqref{eq:integerwavefunction} under the gCST simplifies to:
\begin{align}
&U_{1,s}\left(\vartheta\begin{bmatrix}
\frac{N-1}{2}\\
\frac{N-1}{2}
\end{bmatrix}(Z,\tau)\prod_{1\leq i,j\leq N}\theta_{11}(z_{i}-z_{j},\tau)e^{i\pi\tau N_{\phi}\sum_{j=1}^{N}y_{j}^{2}}\right) =  \nonumber\\
&=\vartheta\begin{bmatrix}
\frac{N-1}{2}\\
\frac{N-1}{2}
\end{bmatrix}(Z_{s},\tau_{s})\prod_{1\leq i,j\leq N}\theta_{11}((z_{i}-z_{j})_{s},\tau_s)e^{i\pi\tau_{s} N\sum_{j=1}^{N}(y_{j})_{s}^{2}}.   
\label{eq:evolutionintegerwihouthalfform}
\end{align}

Combining the results from the IQHE evolution with the gCST in Eq.\eqref{eq:evolutionintegerwihouthalfform} and the half-form evolution in Eq.~\eqref{eq:evolutionhalfformIQHE}, we have:
\begin{align*}
&U_{s}\left(\vartheta\begin{bmatrix}
\frac{N-1}{2}\\
\frac{N-1}{2}
\end{bmatrix}(Z,\tau)\prod_{1\leq i,j\leq N}\theta_{11}(z_{i}-z_{j},\tau)e^{i\pi\tau N_{\phi}\sum_{j=1}^{N}y_{j}^{2}}\otimes \sqrt{\bigwedge_{i=1}^{N} dz_{i}}\right) =\\
&=\vartheta\begin{bmatrix}
\frac{N-1}{2}\\
\frac{N-1}{2}
\end{bmatrix}(Z_{s},\tau_{s})\prod_{1\leq i,j\leq N}\theta_{11}((z_{i}-z_{j})_{s},\tau_s)e^{i\pi\tau N_{\phi}\sum_{j=1}^{N}(y_{j})_{s}^{2}} \otimes \sqrt{\bigwedge_{i=1}^{N}d{z_s}_{i}}.     
\end{align*}

This result aligns with previous findings in the literature on the IQHE on a torus for different modular parameter $\tau$. Therefore, we obtain that, starting with a given flat geometry on the torus $X$, given by some initial choice of modular parameter $\tau$, we do obtain the correct states for the IQHE, for any other value of the modular parameter, by using the geometric quantization inspired evolution along the deformation of geometry given by the gCST. 

The limit $s\rightarrow \infty$ is also insightful. Starting with the evolved section of $L^{\boxtimes N}\otimes K^{1/2}$ and considering $\tau = 0$, we observe that as $s\rightarrow \infty$, the evolution of the wavefunctions with the half-form correction becomes:
\begin{align*}
&\lim_{s\rightarrow \infty} U_{s}\left(\vartheta\begin{bmatrix}
\frac{N-1}{2}\\
\frac{N-1}{2}
\end{bmatrix}(Z,\tau)\prod_{1\leq i,j\leq N}\theta_{11}(z_{i}-z_{j},\tau)e^{i\pi\tau N_{\phi}\sum_{j=1}^{N}y_{j}^{2}}\otimes \sqrt{\bigwedge_{i=1}^{N} dz_{i}}\right) \\
&=\sum_{n\in\mathbb{Z}}\sqrt{\frac{1}{2^{N}N}}\delta\left(Y+n+\frac{N-1}{2}\right)e^{2\pi i\left(n+\frac{N-1}{2}\right)\left(X+\frac{N-1}{2}\right)}\sqrt{dY} \\
&\times \prod_{1\leq i,j\leq N}\sum_{n\in\mathbb{Z}}\sqrt{\frac{1}{2^{N}N}}\delta\left(y_{i}-y_{j}+n+\frac{1}{2}\right)e^{2\pi i\left(n+\frac{1}{2}\right)\left(x_{i}-x_{j}+\frac{1}{2}\right)}\otimes \sqrt{d(y_{i}-y_{j})}. 
\end{align*}
Though more complex than the single-particle case, this result still exhibits Dirac delta distributions, centered at $n+\frac{N-1}{2}$ for the center of mass and $n+\frac{1}{2}$ for each single-particle function $\theta_{11}((z_{i}-z_{j})_{s},\tau_s)$, with $n\in\mathbb{Z}$. 

To better understand the density profile of the integer quantum Hall effect (IQHE), we begin by computing the normalization of the lowest Landau level (LLL) wavefunctions for the IQHE, specifically normalizing Eq.\eqref{eq:integerwavefunction}. With these normalized wavefunctions, we can then calculate the particle density using Eq.\eqref{densitysecondquantization}. Next, we explore the limit $s\rightarrow \infty$ in the density plots. By using a purely imaginary $\tau$ and varying it, we achieve a deformation equivalent to changing $s$, since $\tau_{2_s}=\frac{s}{2\pi N}$. Thus, in the plots below, we vary $\tau_{2_s}$ to observe this effect. From a numerical perspective, an efficient way to examine the particle density in the IQHE is to sum the densities of the normalized one-particle states:
\begin{equation}
\rho_{\mathrm{Integer}}(z) = \sqrt{2N\tau_{2_s}}\sum_{l=0}^{N-1}\Big|\theta_{l/N}(Nz_{s},N\tau_{s})e^{i\pi N\tau_{s} y^2}\Big|^2. 
\label{densidadeinteger}
\end{equation}
The goal of this section is to investigate how the density profile varies with changes in the number of particles and in $\tau_{s}$. First, we examine the effect of changing the particle number. In Fig.\ref{fig:changeNIQHE}, we fix the torus’s modular parameter $\tau_{s}=i$ and observe how the density plots evolve as the number of particles increases.

\begin{figure}[htbp]
  \centering

  \begin{minipage}[b]{0.32\linewidth}
    \centering
    \includegraphics[width=0.9\linewidth]{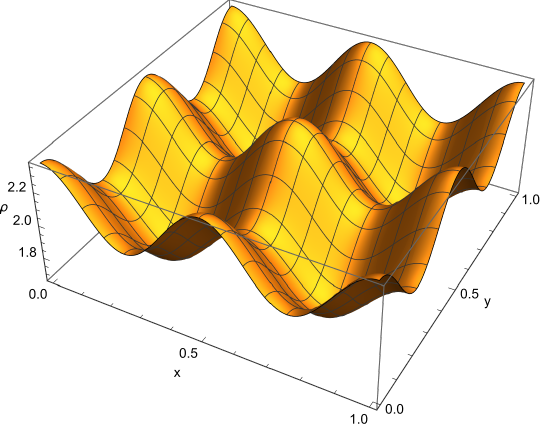}\\[-0.5ex]
    {\small (a) $N=2$}
  \end{minipage}\hfill
  \begin{minipage}[b]{0.32\linewidth}
    \centering
    \includegraphics[width=0.9\linewidth]{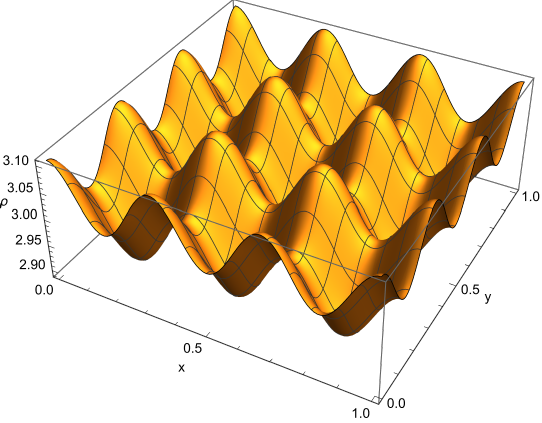}\\[-0.5ex]
    {\small (b) $N=3$}
  \end{minipage}\hfill
  \begin{minipage}[b]{0.32\linewidth}
    \centering
    \includegraphics[width=0.9\linewidth]{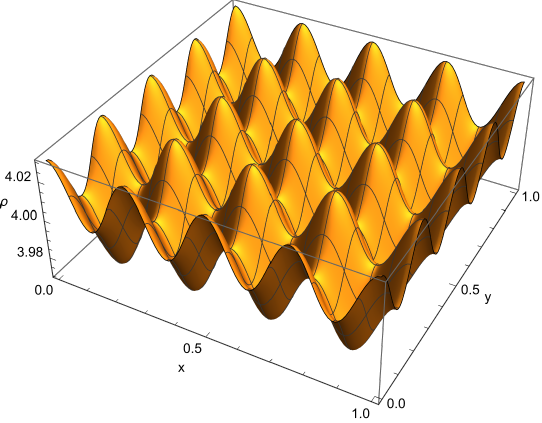}\\[-0.5ex]
    {\small (c) $N=4$}
  \end{minipage}\\[0.5cm]

  \begin{minipage}[b]{0.32\linewidth}
    \centering
    \includegraphics[width=0.9\linewidth]{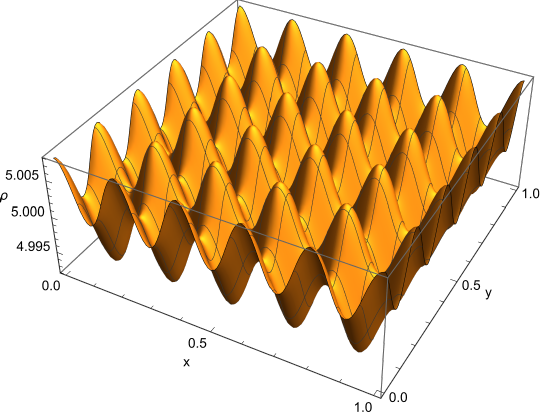}\\[-0.5ex]
    {\small (d) $N=5$}
  \end{minipage}\hfill
  \begin{minipage}[b]{0.32\linewidth}
    \centering
    \includegraphics[width=0.9\linewidth]{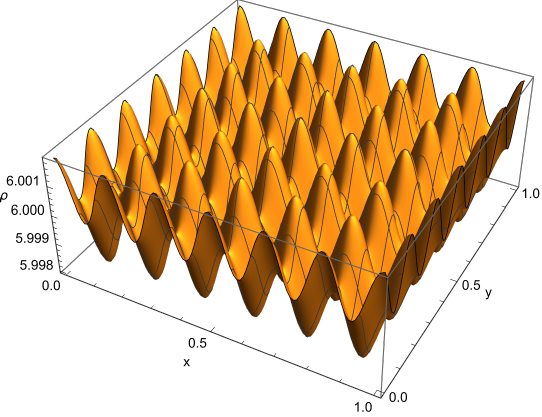}\\[-0.5ex]
    {\small (e) $N=6$}
  \end{minipage}\hfill
  \begin{minipage}[b]{0.32\linewidth}
    \centering
    \includegraphics[width=0.9\linewidth]{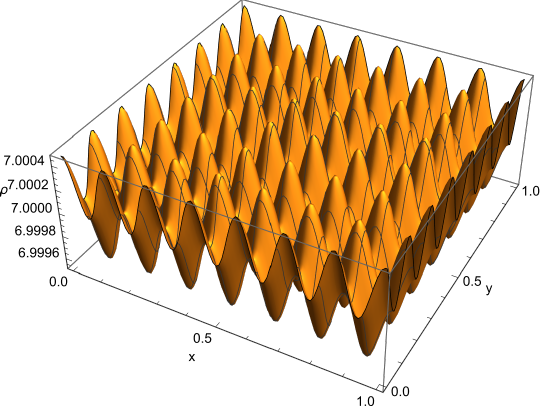}\\[-0.5ex]
    {\small (f) $N=7$}
  \end{minipage}\\[0.5cm]

  \begin{minipage}[b]{0.32\linewidth}
    \centering
    \includegraphics[width=0.9\linewidth]{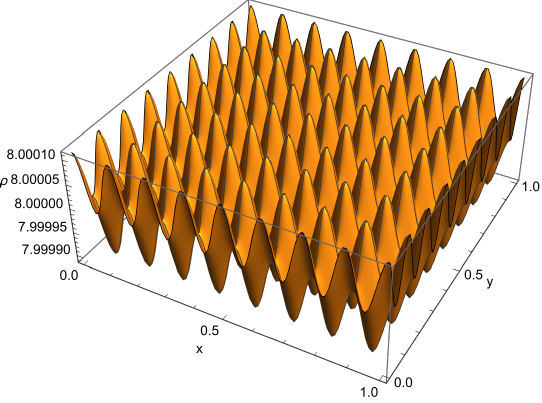}\\[-0.5ex]
    {\small (g) $N=8$}
  \end{minipage}\hfill
  \begin{minipage}[b]{0.32\linewidth}
    \centering
    \includegraphics[width=0.9\linewidth]{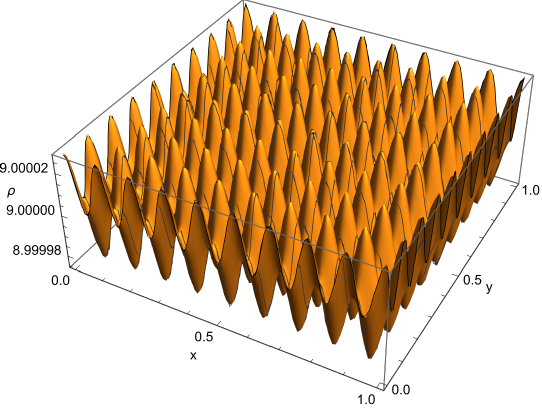}\\[-0.5ex]
    {\small (h) $N=9$}
  \end{minipage}\hfill
  \begin{minipage}[b]{0.32\linewidth}
    \centering
    \includegraphics[width=0.9\linewidth]{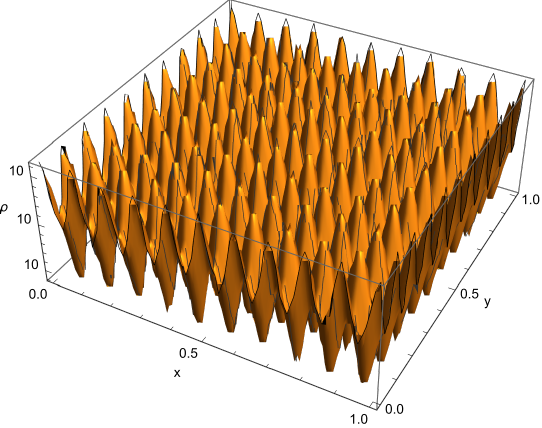}\\[-0.5ex]
    {\small (i) $N=10$}
  \end{minipage}

  \caption{Varying the number of particles while keeping \(\tau = i\).}
  \label{fig:changeNIQHE}
\end{figure}

As seen in Fig.\ref{fig:changeNIQHE}, the number of spikes in the density profile increases with the number of particles. With more particles, the density profile gradually approaches uniformity, a behavior consistent with the findings in \cite{LargeLimitKlevtsov}.

Next, we vary $\tau_{s}$ while keeping the particle number fixed at $N=2$. The resulting density profiles are shown in Fig.\ref{fig:changetauIQHE}.
\begin{figure}[htbp]
  \centering

  \begin{minipage}[b]{0.32\linewidth}
    \centering
    \includegraphics[width=0.8\linewidth]{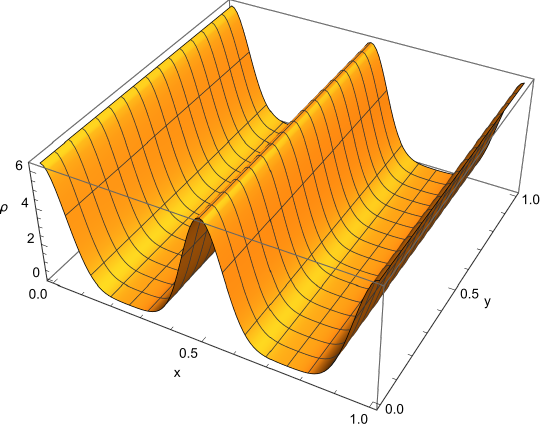}\\[-0.5ex]
    {\small (a) \(\tau_{s}=0.1\,i\)}
  \end{minipage}\hfill
  \begin{minipage}[b]{0.32\linewidth}
    \centering
    \includegraphics[width=0.8\linewidth]{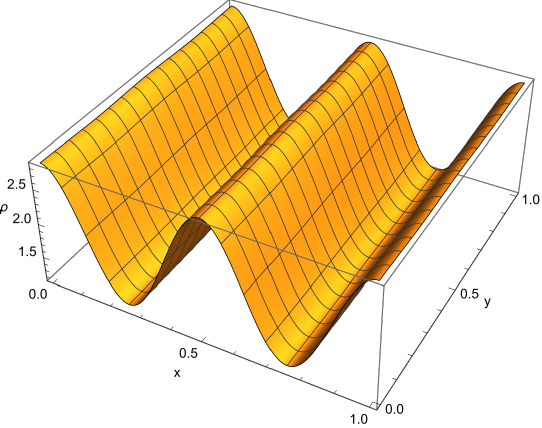}\\[-0.5ex]
    {\small (b) \(\tau_{s}=0.5\,i\)}
  \end{minipage}\hfill
  \begin{minipage}[b]{0.32\linewidth}
    \centering
    \includegraphics[width=0.8\linewidth]{Figures/IQHE-tau=I-Ne=2.pdf}\\[-0.5ex]
    {\small (c) \(\tau_{s}=1\,i\)}
  \end{minipage}\\[0.5cm]

  \begin{minipage}[b]{0.32\linewidth}
    \centering
    \includegraphics[width=0.8\linewidth]{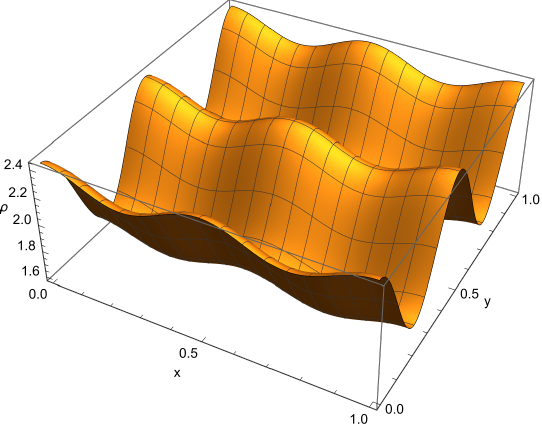}\\[-0.5ex]
    {\small (d) \(\tau_{s}=1.3\,i\)}
  \end{minipage}\hfill
  \begin{minipage}[b]{0.32\linewidth}
    \centering
    \includegraphics[width=0.8\linewidth]{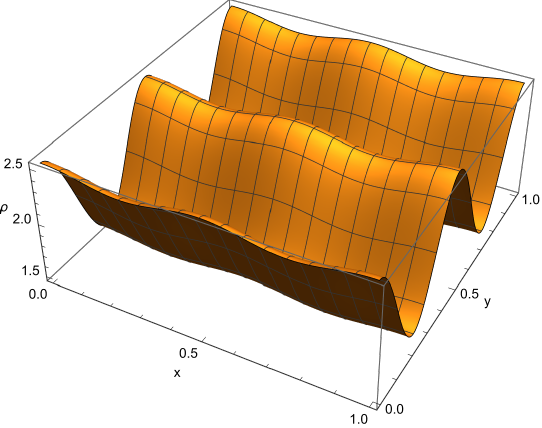}\\[-0.5ex]
    {\small (e) \(\tau_{s}=1.5\,i\)}
  \end{minipage}\hfill
  \begin{minipage}[b]{0.32\linewidth}
    \centering
    \includegraphics[width=0.8\linewidth]{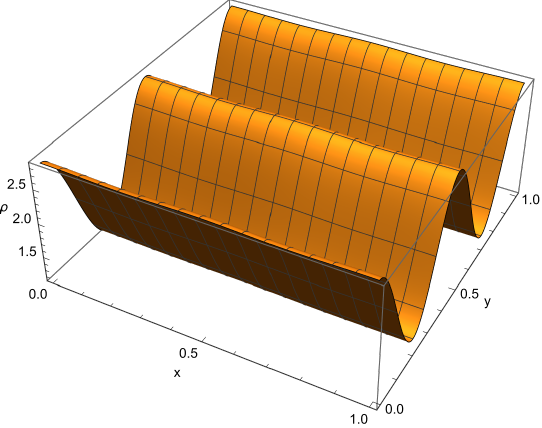}\\[-0.5ex]
    {\small (f) \(\tau_{s}=2\,i\)}
  \end{minipage}\\[0.5cm]

  \begin{minipage}[b]{0.32\linewidth}
    \centering
    \includegraphics[width=0.8\linewidth]{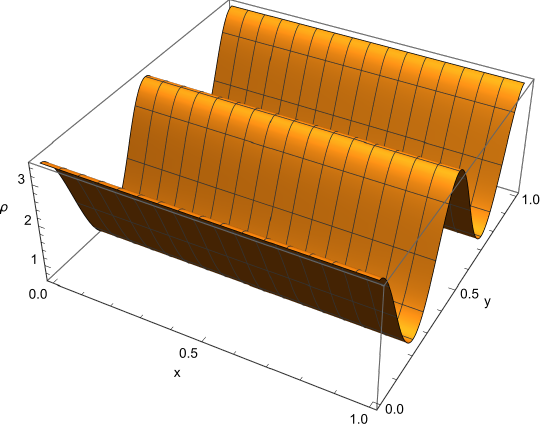}\\[-0.5ex]
    {\small (g) \(\tau_{s}=3\,i\)}
  \end{minipage}\hfill
  \begin{minipage}[b]{0.32\linewidth}
    \centering
    \includegraphics[width=0.8\linewidth]{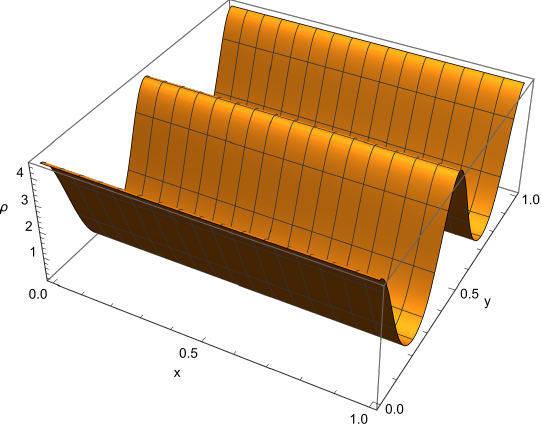}\\[-0.5ex]
    {\small (h) \(\tau_{s}=5\,i\)}
  \end{minipage}\hfill
  \begin{minipage}[b]{0.32\linewidth}
    \centering
    \includegraphics[width=0.8\linewidth]{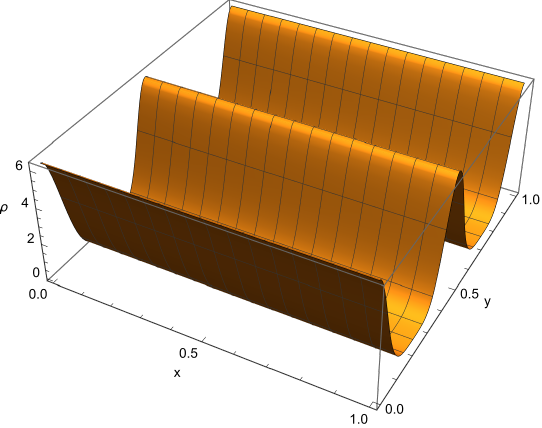}\\[-0.5ex]
    {\small (i) \(\tau_{s}=10\,i\)}
  \end{minipage}

  \caption{Varying \(\tau_{s}\) while keeping \(N=2\).}
  \label{fig:changetauIQHE}
\end{figure}

Several key observations arise from this evolution. First, the density profile for $\tau_{s}=i$ displays a unique symmetry absent in other $\tau_{s}$ values. This symmetry corresponds to $\tau=i$ being a fixed point of the modular transformation group, represented by an $\mathrm{SL}(2,\mathbb{Z})$ matrix.

The behavior of the density maxima aligns with the analytic expression for the density of one-particle states in Eq.\eqref{integrateddensitysingleparticle} and the relationship between IQHE particle density and one-particle state density in Eq.\eqref{densidadeinteger}. Specifically, for very small $|\tau_{s}|$, the oscillatory component along the $x$-direction is more prominent, placing the density maxima along the $x$-axis. As $|\tau_{s}|$ increases, the $y$-dependent oscillatory component becomes more influential, reducing the prominence of $x$-direction oscillations. This is supported by the fact that, analytically, for the one particle state the density of particles converged to Dirac delta distributions on the points $y=n+\frac{l}{N}$ with $n\in\mathbb{Z}$. While $x$-direction oscillations remain noticeable for $\tau_{s}=1.3 i$ and $\tau_{s}=1.5 i$, they appear less pronounced graphically for larger $|\tau_{s}|$. However, computational checks confirm their continued presence. 

As explained above, from the geometric quantization point of view, the limit $s\to\infty$ corresponds to a degeneration of the K\"ahler polarizations of $X$ to a real polarization whose polarized wave functions are distributional and supported on an appropriate collection of Bohr-Sommerfeld cycles.

\subsection{Fractional Quantum Hall Effect}
\label{section:fractionalquantumhalleffect}

Finally, we construct a Laughlin state with $N_{e}$ electrons and a filling fraction $\nu=\frac{N_{e}}{N_{\phi}}=\frac{1}{k}$, where $k$ is an odd number. This state is an element of $\Lambda^{N_{e}} H^{0}(E_{\tau}, L^{kN_{e}})$, that is, $\Psi_{\mathrm{Laughlin}} \in \Lambda^{N_{e}} H^{0}(E_{\tau}, L^{kN_{e}})$. If we require $\Psi_{\mathrm{Laughlin}} \in \Lambda^{N_{e}} H^{0}(E_{\tau}, L^{kN_{e}})$, together with the same vanishing order when the coordinates coallesce as in the plane, we are lead to the following form for the analog of Laughlin's wavefunctions on $E_{\tau}$~\cite{Haldane-Rezayi,LargeLimitKlevtsov}:
\begin{equation*}
    \Psi_{\mathrm{Laughlin}} = f(z_{1},\dots,z_{N})\prod_{1\leq i,j\leq N_{e}}(\theta_{11}(z_{i}-z_{j},\tau))^{k}e^{i\pi \tau k N_{e}\sum_{j=1}^{N_e}y^2_j}. 
\end{equation*}

As in the IQHE, next we ensure that $\Psi_{\mathrm{Laughlin}}$ transforms correctly as an element of $\Lambda^{N_{e}}H^{0}(E_{\tau},L^{kN_{e}})$, requiring it to transform, under lattice translations of each particle separately, with the multipliers $\widetilde{e}^{kN_{e}}_{\gamma}(z)$ that determine the line bundle $L^{kN_e}$. Following the reasoning in Section \ref{section:integerquantumhalleffect}, we conclude that $f$ behaves as the pullback under $g$ of a holomorphic section of $L^{k}\to E_{\tau}$ twisted by a flat line bundle, much like what was found in the case of the IQHE. In fact, the function $f$ satisfies the same quasiperiodicity relations established in Section \ref{sub:theta_functions} for theta functions with characteristics. Consequently, there are exactly $k$ independent holomorphic functions given by
\begin{equation*}
f(z_{1},...,z_{N_e}) = \vartheta\begin{bmatrix}
\frac{N_e-1}{2k}+\frac{l}{k}\\
\frac{N_e-1}{2}
\end{bmatrix}(kZ,k\tau),
\end{equation*}
with $l=0,...,k-1$. 
Physically, this indicates that the ground state of the system---presumably adiabatically connected to the present Laughlin state---is $k^g$-fold degenerate on a surface of genus $g$~\cite{wen:niu:1990}, where in our case $g=1$. This degeneracy is, by definition, a manifestation of topological order. Additionally, the ground state degeneracy is related to the fact that, at filling fraction $1/k$, the system's relevant excitations are quasiparticles with a charge of $1/k$ of an electron. Moving $k$ such quasiparticles to the same spatial point would yield a full electron, reflecting that these quasiparticles do not obey Fermi-Dirac statistics.


We conclude that the Laughlin wavefunctions exhibit $k$-fold degeneracy, with a basis given by
\begin{equation}
    \Psi_{\mathrm{Laughlin}}(z_{1},...,z_{N_e})=\vartheta\begin{bmatrix}
\frac{N_e-1}{2k}+\frac{l}{k}\\
\frac{N_e-1}{2}
\end{bmatrix}(kZ,k\tau)\prod_{1\leq i,j\leq N_{e}}(\theta_{11}(z_{i}-z_{j},\tau))^{k}e^{i\pi \tau k N_{e}\sum_{j=1}^{N_e}y^2_j},
\label{eq:Laughlinfractional}
\end{equation}
with $l=0,...,k-1$.

\subsection{Evolution of Laughlin states under the gCST and the Tao-Thouless state}
\label{subsec:EvolutionFQHE}

Given the similarities between the wavefunction of the LLL in Eq.~\eqref{eq:Laughlinfractional} and the IQHE wavefunction in Eq.~\eqref{eq:integerwavefunction}, the evolution of the LLL wavefunction can be directly inferred from that of the IQHE. Thus, we will simply present the final result here and refer the reader to Section \ref{subsec:evolutionintegerquantumhalleffect} for further details.

The trivializing section of the square root of the canonical line bundle $K^{1/2}$ in this case is $\sqrt{dz_{1}\wedge...\wedge dz_{N_e}}$. As explained in previous sections, the gCST can be decomposed into the tensor product of two components: $U_{1,s}$ and  $U_{2,s}$, where the later acts on the half-form bundle. In Section \ref{subsec:evolutionintegerquantumhalleffect}, we found that the action of $U_{2,s}$ on the half-form is just given by analytic continuation 
\begin{equation}
    U_{2,s}(\sqrt{dz_{1}\wedge...\wedge dz_{N_e}}) = \sqrt{dz_{s_1}\wedge...\wedge dz_{s_Ne}}. 
    \label{evolutionhalfformFQHE}
\end{equation}
Additionally, as argued in Sec.\ref{subsec:evolutionintegerquantumhalleffect}, $U_{1,s}$ acts on the LLL wavefunction in Eq.~\eqref{eq:Laughlinfractional} by transforming the coordinates $z_i-z_j$, $\tau$ and $Z$ according to:
\begin{align*}
&\tau\rightarrow \tau_{s} = \tau + \frac{is}{2\pi kN_{e}},\\
&z_{i}-z_{j}\rightarrow (z_{i}-z_{j})_{s} = z_{i}-z_{j}+\frac{is}{2\pi kN_{e}}(y_{i}-y_{j}),
&Z\rightarrow Z_{s} = Z+\frac{isY}{2\pi kN_{e}}. 
\end{align*}

Thus, the evolution is given by:
\begin{align}
&U_{1,s}\left(\vartheta\begin{bmatrix}
\frac{N_e-1}{2k}+\frac{l}{k}\\
\frac{N_e-1}{2}
\end{bmatrix}(kZ,k\tau)\prod_{1\leq i,j\leq N_e}\left(\theta_{11}(z_{i}-z_{j},\tau)\right)^ke^{i\pi\tau k N_{e}\sum_{j=1}^{N_e}y_{j}^{2}}\right) =  \nonumber\\
&=\vartheta\begin{bmatrix}
\frac{N_e-1}{2k}+\frac{l}{k}\\
\frac{N_e-1}{2}
\end{bmatrix}(kZ_s,k\tau_s)\prod_{1\leq i,j\leq N_e}\left(\theta_{11}((z_{i}-z_{j})_{s},\tau_s)\right)^ke^{i\pi\tau_{s} k N_e\sum_{j=1}^{N_e}(y_{j})_{s}^{2}}.   
\label{eq:evolutionFQHEwihouthalfform}
\end{align}

Combining the gCST evolution for the FQHE wavefunction in Eq.\eqref{eq:evolutionFQHEwihouthalfform} with the half-form evolution in Eq.~\eqref{evolutionhalfformFQHE}, we obtain:
\begin{align*}
&U_{s}\left(\vartheta\begin{bmatrix}
\frac{N_e-1}{2k}+\frac{l}{k}\\
\frac{N_e-1}{2}
\end{bmatrix}(kZ,k\tau)\prod_{1\leq i,j\leq N_e}\left(\theta_{11}(z_{i}-z_{j},\tau)\right)^k e^{i\pi\tau k N_{e}\sum_{j=1}^{N_e}y_{j}^{2}}\otimes \sqrt{\bigwedge_{i=1}^{N_e} dz_{i}}\right) =\\
&=\vartheta\begin{bmatrix}
\frac{N_e-1}{2k}+\frac{l}{k}\\
\frac{N_e-1}{2}
\end{bmatrix}(kZ_s,k\tau_s)\prod_{1\leq i,j\leq N_e}\left(\theta_{11}((z_{i}-z_{j})_{s},\tau_s)\right)^k e^{i\pi\tau k N_{e}\sum_{j=1}^{N_e}(y_{j})_{s}^{2}} \otimes \sqrt{\bigwedge_{i=1}^{N_e}d{z_s}_{i}}.     
\end{align*}
This result aligns with previous findings in the literature on the Laughlin states on a torus for different modular parameter $\tau$. Thus, also for the Laughlin states of the FQHE on the torus, the correct dependence on the modular parameter is successfully described by the natural evolution of states along deformations of K\"ahler geometry which is given by the gCST.

The limit $s\rightarrow \infty$ is similar to that obtained for the IQHE in Section \ref{subsec:evolutionintegerquantumhalleffect}, with the difference that the locations of the Dirac delta distributions shift due to the additional $l/k$ factor in the fractional quantum Hall states in Eq.~\eqref{eq:Laughlinfractional}.
\begin{align*}
&\lim_{s\rightarrow \infty} U_{s}\left(\vartheta\begin{bmatrix}
\frac{N_e-1}{2k}+\frac{l}{k}\\
\frac{N_e-1}{2}
\end{bmatrix}(kZ,k\tau)\prod_{1\leq i,j\leq N_e}(\theta_{11}(z_{i}-z_{j},\tau))^{k}e^{i\pi\tau k N_{e}\sum_{j=1}^{N_e}y_{j}^{2}}\otimes \sqrt{\bigwedge_{i=1}^{N_e} dz_{i}}\right) \\
&=\sum_{n\in\mathbb{Z}}\sqrt{\frac{1}{2^{N_e}N_{e}k}}\delta\left(Y+n+\frac{N_{e}-1}{2k}+\frac{l}{k}\right)e^{2\pi i\left(n+\frac{N_{e}-1}{2k}+\frac{l}{k}\right)\left(X+\frac{N-1}{2}\right)}\otimes \sqrt{dY}.\\
&\times \prod_{1\leq i,j\leq N_{e}}\left(\sum_{n\in\mathbb{Z}}\sqrt{\frac{1}{2^{N_e}N_{e}k}}\delta\left(y_{i}-y_{j}+n+\frac{1}{2}\right)e^{2\pi ik\left(n+\frac{1}{2}\right)\left(x_{i}-x_{j}+\frac{1}{2}\right)}\right)^k \otimes \sqrt{d(y_{i}-y_{j})}. 
\end{align*}

We now examine how the density profile of Laughlin states evolves as the torus geometry is deformed, i.e., as $s\rightarrow\infty$. To achieve this, we compute the particle density using Eq.\eqref{densitysecondquantization}. Instead of directly changing $s$, we vary the purely imaginary $\tau_{s}$, which produces equivalent results at the density level since $\tau_{2,s} = s/(2\pi k N_e)$. For this analysis, we focus on the FQHE at filling fraction $\nu=1/3$. Numerically, we explore two scenarios: keeping $\tau_{s}$ fixed while varying the particle number, and keeping the particle number fixed while varying $\tau_{2,s}$.

Starting with $\tau_{s}=i$ fixed, we vary the number of particles. As the normalization integral and particle density calculation become challenging for large $N_e$, we compute these only for $N_e=2$ and $N_e=3$. The results are shown in Fig.\ref{fig:changeNFQHE}. Next, we vary $\tau_{s}$ while keeping $N_e=2$ constant, with the results displayed in Fig.\ref{fig:changetauFQHE}. 

\begin{figure}[htbp]
  \centering
  \begin{minipage}[b]{0.48\linewidth}
    \centering
    \includegraphics[width=0.6\linewidth]{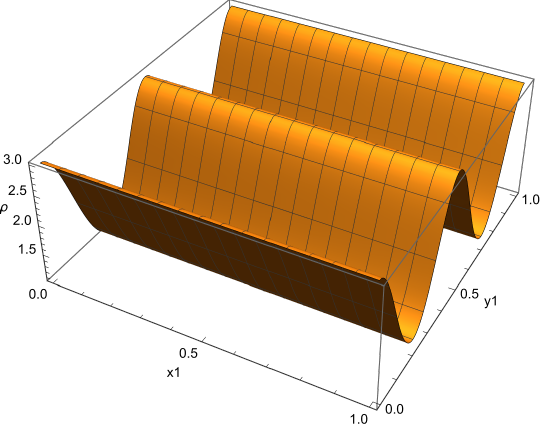}\\[-0.5ex]
    {\small (a) \(N_{e}=2\)}
  \end{minipage}\hfill
  \begin{minipage}[b]{0.48\linewidth}
    \centering
    \includegraphics[width=0.6\linewidth]{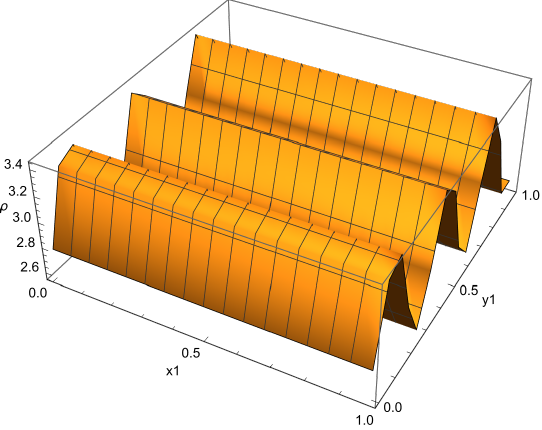}\\[-0.5ex]
    {\small (b) \(N_{e}=3\)}
  \end{minipage}

  \caption{Varying the number of particles of the system while keeping \(\tau = i\).}
  \label{fig:changeNFQHE}
\end{figure}

\begin{figure}[htbp]
  \centering

  \begin{minipage}[b]{0.32\linewidth}
    \centering
    \includegraphics[width=0.8\linewidth]{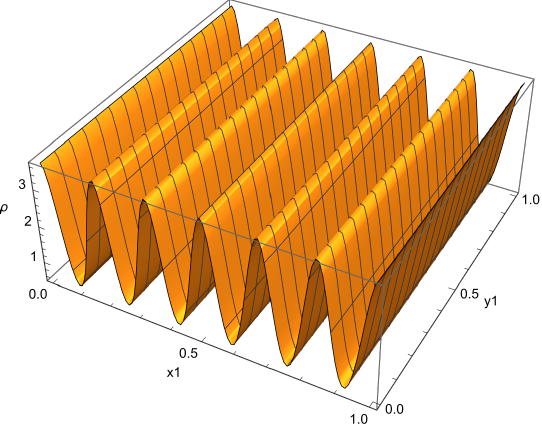}\\[-0.5ex]
    {\small (a) \(\tau_{s}=0.1\,i\)}
  \end{minipage}\hfill
  \begin{minipage}[b]{0.32\linewidth}
    \centering
    \includegraphics[width=0.8\linewidth]{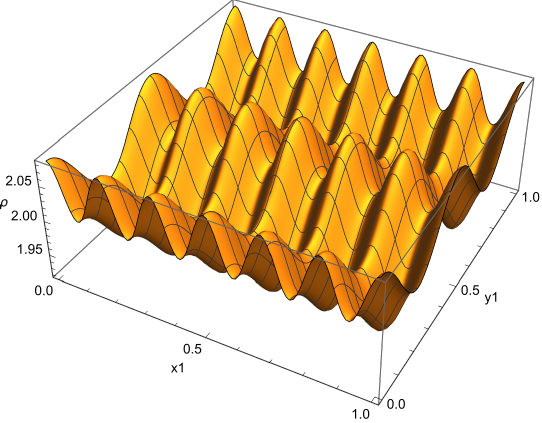}\\[-0.5ex]
    {\small (b) \(\tau_{s}=0.5\,i\)}
  \end{minipage}\hfill
  \begin{minipage}[b]{0.32\linewidth}
    \centering
    \includegraphics[width=0.8\linewidth]{Figures/FQHE-1.3-Ne=2-tau=1.pdf}\\[-0.5ex]
    {\small (c) \(\tau_{s}=1\,i\)}
  \end{minipage}\\[0.5cm]

  \hspace*{0.18\linewidth}%
  \begin{minipage}[b]{0.32\linewidth}
    \centering
    \includegraphics[width=0.8\linewidth]{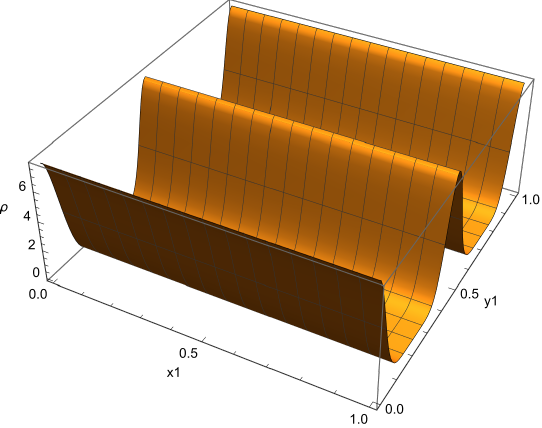}\\[-0.5ex]
    {\small (d) \(\tau_{s}=5\,i\)}
  \end{minipage}\hfill
  \begin{minipage}[b]{0.32\linewidth}
    \centering
    \includegraphics[width=0.8\linewidth]{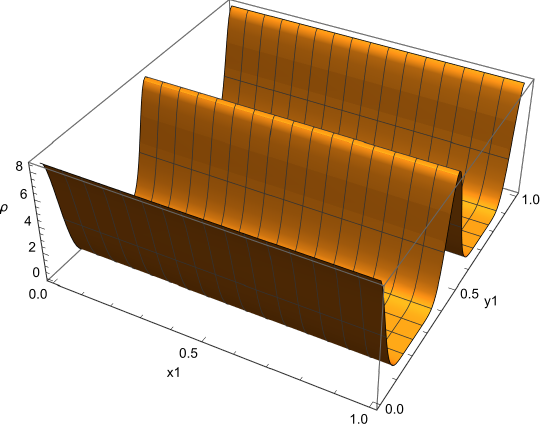}\\[-0.5ex]
    {\small (e) \(\tau_{s}=6\,i\)}
  \end{minipage}

  \caption{Varying \(\tau_{s}\) while keeping \(N_{e}=2\).}
  \label{fig:changetauFQHE}
\end{figure}

In conclusion, the density profile results for the FQHE closely resemble those for the IQHE, as analyzed in Section \ref{subsec:evolutionintegerquantumhalleffect}. The primary difference is the symmetry observed at $\tau_{s}=i$ for the IQHE, which is absent in the FQHE.

The gCST describes, as $s\to\infty$,  an adiabatic deformation of the Laughlin state on the torus along the deformation of the geometry. In this limit,  the geometry corresponds to a thin torus, given by a product of a circle of radius $\to 0$ and a circle of radius $\to \infty$ such that the symplectic area remains fixed through the deformation. The FQHE for this limit geometry is described by the Tao-Thouless state \cite{PhysRevB.77.155308, hansson2009taothoulessrevisited} where electrons localize on specific one-dimensional cycles on the torus which separated by integral symplectic areas. These cycles correspond to the Bohr-Sommerfeld loci where the one-particle states localize in this singular geometric limit and, in terms of geometric quantization, correspond to the support of the distributional polarized sections for the limit real polarization of the torus. We note also that in \cite{Zhou_Nussinov_Seidel} the heat equation satisfied by theta functions was used to study the deformation of the Laughlin states on the torus for varying modular parameter. The heat kernel is, of course, one of the essential ingredients of the quantum operator in the gCST 
(see (\ref{eq:GCST1})). Together with the results in the next section, which addresses the QHE for 
flat geometries on the plane, these results indicate that indeed the geometric quantization prescription for evolution of quantum states under deformation of geometry---the gCST---describes correctly the effect of such deformations on the quantum wave functions that are relevant for the QHE. For the case of the application of the gCST to the FQHE on the sphere see \cite{matos:mera:mourao:mourao:nunes:2023}.


\section{Non-Flat Geometry and Periodic Hamiltonian}
\label{sec:non_flat_geometry}

In this section, we apply the gCST framework to the quantum Hall effect in non-flat geometries of the 2-torus. Previously, the focus was on flat geometries with zero Gauss curvature, corresponding  to the usual model of the torus obtained through a quotient of the  Euclidean plane by a rank 2 lattice defined by the modular parameter $\tau$. Introducing a periodic Hamiltonian generates nonzero Gauss curvature, causing the torus geometry to deviate locally from the Euclidean plane geometry. We find explicitly the  geometric deformations and their impact on the wavefunctions of the quantum Hall system.

The starting point of a flat K\"ahler
structure on $X=E_\tau$ and the quantum states are the same 
as in section \ref{sec:flat_geometry},
but instead of deforming 
with the imaginary time Hamiltonian flow
of $H=\frac{y^2}{2}$ we use a global  real analytic function, H, on $E_\tau$.
The gCST framework is adapted to include the periodic Hamiltonian, with the 
prequantum and quantum operators changing accordingly.

In Section \ref{section:oneparticlestatesnonflat}, we introduce the new pre-quantum and quantum operators and apply the gCST to the single-particle wavefunctions from Section \ref{section:oneparticlestates}. This approach is then extended to multiparticle states. Section \ref{section:integerquantumhalleffectnonflat} focuses on evolving the integer quantum Hall effect wavefunctions obtained in Section \ref{section:integerquantumhalleffect} under the revised gCST. Finally, Section \ref{section:fractionalquantumhalleffectnonflat} describes the evolution of fractional quantum Hall effect wavefunctions, derived in Section \ref{section:fractionalquantumhalleffect}. (Throughout this section, for simplicity, we do not include the half-form correction.)

\subsection{One-Particle States and evolution under the gCST}
\label{section:oneparticlestatesnonflat}
We recall from Section \ref{section:oneparticlestates} that the symplectic form on the complex torus $E_{\tau}$ can be expressed in coordinates $(x,y)$  as $\omega = 2\pi N_{\phi}dx\wedge dy$.  Since we are working in the unitary gauge, the connection 1-form is given by $A = 2\pi i N_{\phi}y dx$. The Chern connection is given by
\begin{equation}
\nabla_{X}s= X\cdot s+A(X)\cdot s,
\end{equation}
where $X\in \mathfrak{X}(E_{\tau})$ and $s$ denotes sections of the single-particle line bundle. We use the same polarized sections constructed for the unitary line bundle in Sect.~\ref{section:oneparticlestates}. As previously discussed, these correspond to the wavefunctions of the LLL, 
and to the Hilbert space $\mathcal{H}_{P_0}$ of the initial
polarization ${P}_0$
with basis
\begin{equation} 
\label{e-psilll}
\psi^{(l)}_{\mathrm{LLL}}(z,\tau) = \theta_{l/N_{\phi}}(N_{\phi}z, N_{\phi}\tau) e^{i\pi N_{\phi}\tau y^2} \,  , \, l=0, \dots , N_{\phi} -1 . 
\end{equation} 
For simplicity, the half-form correction of the wavefunctions is omitted in this section.


In this section, we illustrate our method by considering the following representative  periodic Hamiltonian
\begin{equation}
    H(y) = \sin^2{(2\pi y)} \, ,
\end{equation}
and the one dimensional family of K\"ahler metrics corresponding to the geodesics with initial  K\"ahler potential of the flat metric, $k_0$, and initial velocity $H$,
i.e.
\begin{equation}
\label{e-IC}
    \left\{
    \begin{array}{rcl}
       k(0,y)  & =  & \pi N_\phi \tau_2 \, \frac{y^2}{2}\\
        \frac{\partial k_0}{\partial s}(0,y)  & = &  H(y) = \sin^2{(2\pi y)} \, . 
    \end{array}
    \right.
\end{equation}

We aim to study the evolution of the wavefunctions (\ref{e-psilll}) under the gCST corresponding to the non-flat geometries 
generated by $H$. As in (\ref{timescst}) the gCST  is given by: \begin{equation} 
U_{s} = e^{s \, Q_{\mathrm{pre}}(H) - s \, Q(H)}. 
\end{equation} 
 The  Hamiltonian vector field $X_{H}$ reads
 $$
 X_{H} = \frac{\sin(4\pi y)}{N_{\phi}}\frac{\partial}{\partial x}
 $$
  and thus the prequantum operator  is:
\begin{equation}
   Q_{\mathrm{pre}}(H) = i\nabla_{X_{H}} + H =  +iX_{H} + iA(X_{H}) + H = iX_{H}- 2\pi N_{\phi}y\frac{\sin(4\pi y)}{N_{\phi}} + H \, . 
\end{equation}

To define the quantum operator 
$$
Q(H)= Q(\sin^2(2\pi y)) \, ,
$$
corresponding to the Hamiltonian $H$, 
we have the difficulty that no Hamiltonian vector field of a, nonconstant,
periodic function preserves the K\"ahler
polarization. In the case of $H$ a quadratic
polynomial in the variables $x, y$, as in the previous section,
the lack of periodicity of $H$ was compensated
by the fact that the complex time flow 
of $H$ preserved the space flat K\"ahler structures.

Here we will need to construct $Q(H)$
with tbe help of 
 the unitary representation of the finite Heisenberg
group (see \cite{Mumford83}, p 5-11) on the space of level $N_\phi$ theta functions,
$\mathcal{H}_{P}$, with three generators,
$$
\left\{S = Q(e^{2 \pi i y}), \, T = Q(e^{2 \pi i x}), \, e^{-2 \pi \frac{i}{N_\phi}}  \right\}   \, . 
$$
The basis vectors $\psi^{(l)}_{\mathrm{LLL}}(z,\tau)$ in (\ref{e-psilll}) are eigenfunctions of $S$,
$$
S \, (\psi^{(l)}_{\mathrm{LLL}}(z,\tau)) = e^{-2 \pi i \, \frac{l}{N_\phi}} \, \psi^{(l)}_{\mathrm{LLL}}(z,\tau) \, .
$$
The minus sign in the exponent of the eigenvalues is due to different conventions from those in \cite{Mumford83}, p. 9.
We use the unitary operator $S$ and the classical relation between 
$H$ and $e^{2 \pi i \, y}$
to define  $Q(H)$,
\begin{equation}
    Q(H) = Q\left(\sin^2(2\pi y)\right) := \left(\frac{S - S^{-1}}{2i}\right)^2 . 
\label{eq:quantumoperatornonflat}
\end{equation}
The basis functions of (\ref{e-psilll}) are then also eigenfunctions of $Q(H)$,
\begin{equation}
\label{e-eigen-QH}
Q(H) \, (\psi^{(l)}_{\mathrm{LLL}}(z,\tau)) = 
\sin^2(2\pi \frac{l}{N_\phi})
 \, \psi^{(l)}_{\mathrm{LLL}}(z,\tau) \, .
    \end{equation}
and we can define the gCST:
\begin{equation}
   U_{s} = e^{sQ_{\mathrm{pre}}(H)}\circ e^{-sQ(H)} = e^{s\sin^2(2\pi y)}e^{-s 2\pi y\sin(4\pi y)}e^{is\frac{\sin(4\pi y)}{N_{\phi}}\frac{\partial}{\partial x}}e^{-s Q(H)}.  
   \label{eq:GCSTnonflatsingleparticle}
\end{equation}

The operators in Eq.~\eqref{eq:GCSTnonflatsingleparticle} commute, so we first apply $e^{is\frac{\sin(4\pi y)}{N_{\phi}}\frac{\partial}{\partial x}}$ to the single-particle wavefunction in Eq.~\eqref{LLLunitary}. 
\begin{align*}
e^{is\frac{\sin(4\pi y)}{N_{\phi}}\frac{\partial}{\partial x}}(\theta_{l/N_{\phi}}(N_{\phi}z,N_{\phi}\tau)e^{i\pi N_{\phi}\tau y^2}) &=  \theta_{l/N_{\phi}}\left(N_{\phi}\left(x+\frac{\sin(4\pi y)s}{N_{\phi}} \, i +\tau y\right),N_{\phi}\tau\right)e^{i\pi N_{\phi}\tau y^2}\\
&= \theta_{l/N_{\phi}}(N_{\phi}z_{is},N_{\phi}\tau)e^{i\pi N_{\phi}\tau y^2},
\end{align*}
where, $z_{is}=z+\frac{\sin(4\pi y)s}{N_{\phi}} \, i,  $ is the evolved coordinate.

Next, we apply the quantum operator $e^{-sQ(H)}$. 
From (\ref{e-eigen-QH}) we see that
\begin{align*}
e^{-sQ(H)}\theta_{l/N_{\phi}}(N_{\phi}z_{s},N_{\phi}\tau)e^{i\pi N_{\phi}\tau y^2}&= e^{-s \sin^2\left(2\pi\frac{l}{N_{\phi}}\right)}\theta_{l/N_{\phi}}(N_{\phi}z_{t},N_{\phi}\tau)e^{i\pi N_{\phi}\tau y^2}. 
\end{align*}

Combining all contributions, the complete evolution of the LLL wavefunction under gCST is:

\begin{align}
& U_{s}\Psi_{\mathrm{LLL}}(z,\tau) = U_{s}\left(\theta_{l/N_{\phi}}(N_{\phi}z,N_{\phi}\tau)e^{i\pi N_{\phi}\tau y^2}\right)\nonumber\\
\label{evolvedsingleparticlenonflat}
&= e^{s\sin^2{(2\pi y)}}e^{-s 2\pi y\sin{(4\pi y)}}e^{-s\sin^2{\left(2\pi\frac{l}{N_{\phi}}\right)}}\theta_{l/N_{\phi}}\left(N_{\phi}\left(z+\frac{is\sin{(4\pi y)}}{N_{\phi}}\right),N_{\phi}\tau\right)e^{i\pi N_{\phi}\tau y^2}. 
\end{align}

We now examine the evolution of the density profile of single-particle states as the torus geometry is deformed to singular geometries, $s \to s_c$. To achieve this, we compute the particle density using Eq.\eqref{densitysecondquantization}. Unlike the flat case, this computation presents additional challenges. First, the normalization of the single-particle wavefunctions in Eq.\eqref{evolvedsingleparticlenonflat} can no longer be obtained analytically, requiring numerical evaluation, which we did using \textit{Mathematica}. Second, the deformation parameter $s$ must be varied carefully to account for its impact on the geometry, particularly the Gauss curvature of the torus. The particle density has physical meaning only when the Gauss curvature $K$ remains finite. Divergences in $K$ correspond to singularities in the geometry, rendering the torus unphysical. To ensure nonsingular and physically meaningful geometries, we derive the Gauss curvature $K$ as a function of the parameter $s$.
The Gauss curvature $K$ is computed using a key result from differential geometry. For a closed, oriented 2-dimensional Riemannian manifold $M$, the \textit{Gauss-Bonnet theorem} relates the integrated curvature to the topology of $M$:
\begin{equation}
    \int_{M} KdA = 2\pi e(TM)\cdot[M],
\end{equation}
where $K$ is the Gauss curvature, $dA$ is the area form, and $e(TM)$ is the Euler class of the tangent bundle $TM$. For the torus $\mathbb{T}^2$, we  equip $T\mathbb{T}^2$ with a Levi-Civita connection compatible with the metric. The holomorphic tangent bundle $(T\mathbb{T}^2)^{(1,0)}$ is an Hermitian line bundle with an induced Hermitian metric, $h$. The curvature of the corresponding Chern connection relates to the Gauss curvature $K$ via:
\begin{equation}
    i\Omega = KdA,
    \label{eq:curvature_gauss_curvature}
\end{equation}
where $i\Omega$ is the curvature form and $dA$ is identified as the symplectic form $\omega$ on the torus. The curvature form can be expressed in terms of the Hermitian metric, $h$, on the bundle $(T\mathbb{T}^2)^{(1,0)}$:
\begin{equation}
    \Omega = \overline{\partial}\partial\log(h).
    \label{eq:curvature_hermitian_metric}
\end{equation}
To compute the Gauss curvature explicitly, we first determine the Hermitian metric of the torus after deformation. The evolved complex coordinate $z_s=x+\tau y+\frac{is \sin{(4\pi y)}}{N_\phi}=x+\tau y + \frac{is H'(y)}{2\pi N_{\phi}}$ from Eq.~\eqref{evolvedsingleparticlenonflat} allows us to express the symplectic form $\omega$ in the following way:
\begin{equation}
\omega = 2\pi N_{\phi}dx\wedge dy  = \frac{2\pi N_{\phi}}{\tau_{2}+\frac{s}{2\pi N_{\phi}}H''(y)}\frac{i}{2}dz_{s}\wedge d\overline{z}_{s} = \frac{i}{2}h dz_{s}\wedge d\overline{z}_{s},
\end{equation}
where $h=\frac{2\pi N_{\phi}}{\tau_{2}+\frac{s}{2\pi N_{\phi}}H''(y)}$. Substituting $h$ into Eq.~\eqref{eq:curvature_hermitian_metric}, the curvature form becomes:
\begin{equation}
i\Omega = \left(\frac{-2}{(4\pi N_{\phi})^2}\frac{\partial^2 h}{\partial y^2}\right)\omega. 
\end{equation}
From Eq.~\eqref{eq:curvature_gauss_curvature}, we extract the Gauss curvature:
\begin{equation}
K = \frac{i\Omega}{\omega} = \frac{2}{(4\pi N_{\phi})^2}\frac{\partial^2 h}{\partial y^2}. 
\end{equation}
This result reveals the dependence of $K$ on $s$. For fixed values of $N_{\phi}=2$ and $\tau = i$, the Gauss curvature as a function of $y$ is illustrated in Fig.~\ref{fig:gauss_torus}:

\begin{figure}[htbp]
  \centering

  \begin{minipage}[b]{0.48\linewidth}
    \centering
    \includegraphics[width=0.8\linewidth]{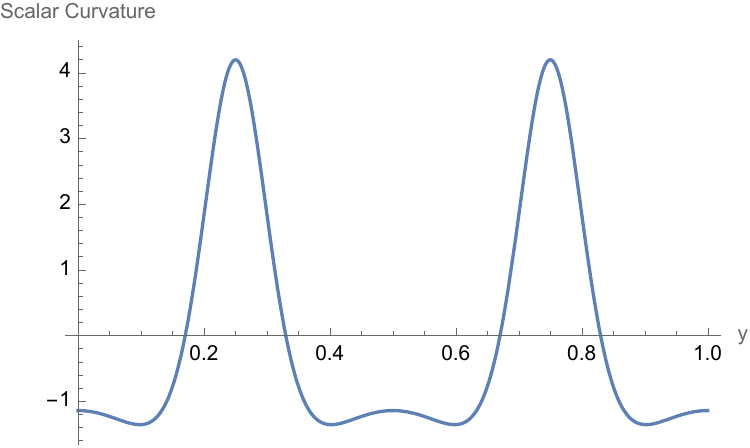}\\[-0.5ex]
    {\small (a) \(s=0.05\)}
  \end{minipage}\hfill
  \begin{minipage}[b]{0.48\linewidth}
    \centering
    \includegraphics[width=0.8\linewidth]{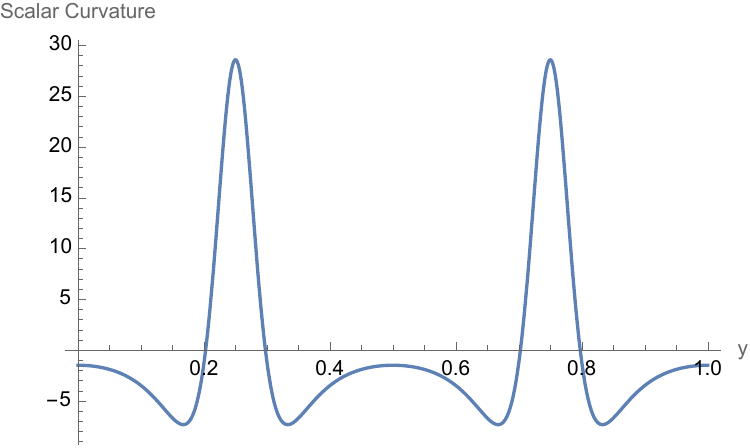}\\[-0.5ex]
    {\small (b) \(s=0.1\)}
  \end{minipage}\\[0.5cm]

  \begin{minipage}[b]{0.48\linewidth}
    \centering
    \includegraphics[width=0.8\linewidth]{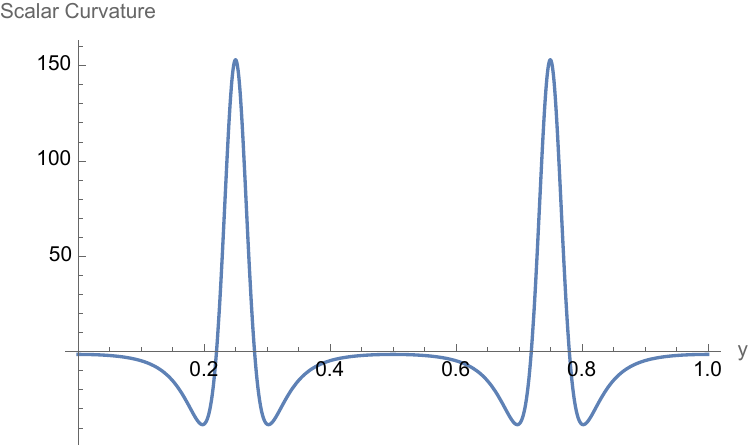}\\[-0.5ex]
    {\small (c) \(s=0.13\)}
  \end{minipage}\hfill
  \begin{minipage}[b]{0.48\linewidth}
    \centering
    \includegraphics[width=0.8\linewidth]{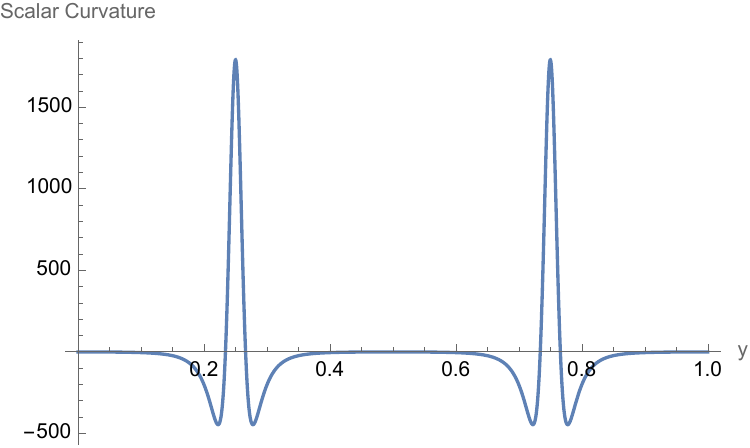}\\[-0.5ex]
    {\small (d) \(s=0.15\)}
  \end{minipage}

  \caption{Gauss curvature of the Torus in terms of the value of \(s\).}
  \label{fig:gauss_torus}
\end{figure}

We can now examine the density profiles of the single-particle states. The Kähler metric $g$ is written as $g= h|dz|^2$, so it can be expressed as:
\begin{equation}
g=\frac{2\pi N_{\phi}}{\tau_{2}+\frac{s}{2\pi N_{\phi}}H''(y)}|dz|^2,
\end{equation}
it is evident that for $y=0$, $N_{\phi}=2$, and $\tau_{2}=1$, the metric diverges at a critical value of $s$, $s_{c}=\frac{1}{2\pi}\sim 0.16$. We refer to $s_{c}$ as the critical value of s. For $s \geq s_{c}$, the polarization is no longer Kähler, and the system becomes nonphysical. Consequently, we restrict our analysis to $s<s_c$, setting an upper limit of $s=0.15$.   

We begin by analyzing the density of the single-particle states while keeping $N_{\phi}$ and $\tau$ fixed and varying $s$:

\begin{figure}[htbp]
  \centering

  \begin{minipage}[b]{0.48\linewidth}
    \centering
    \includegraphics[width=0.8\linewidth]{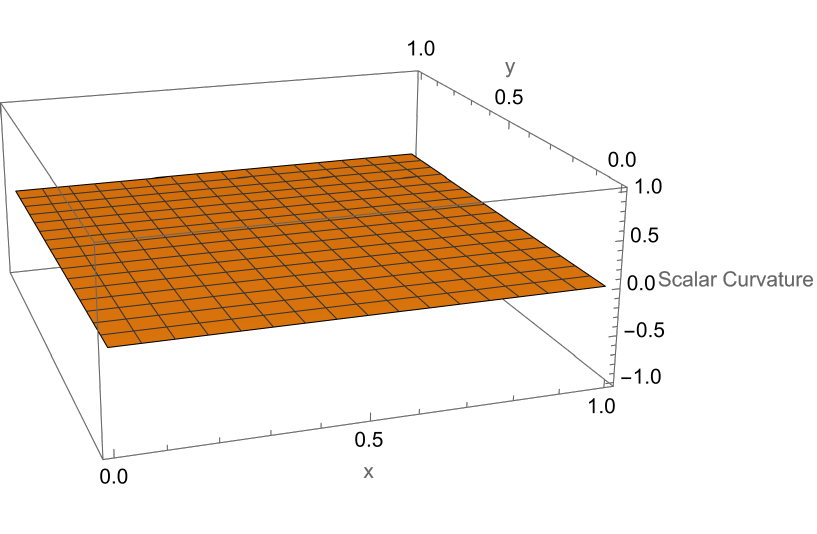}\\[-0.5ex]
    {\small (a) Scalar Curvature \(s=0\)}
  \end{minipage}\hfill
  \begin{minipage}[b]{0.48\linewidth}
    \centering
    \includegraphics[width=0.8\linewidth]{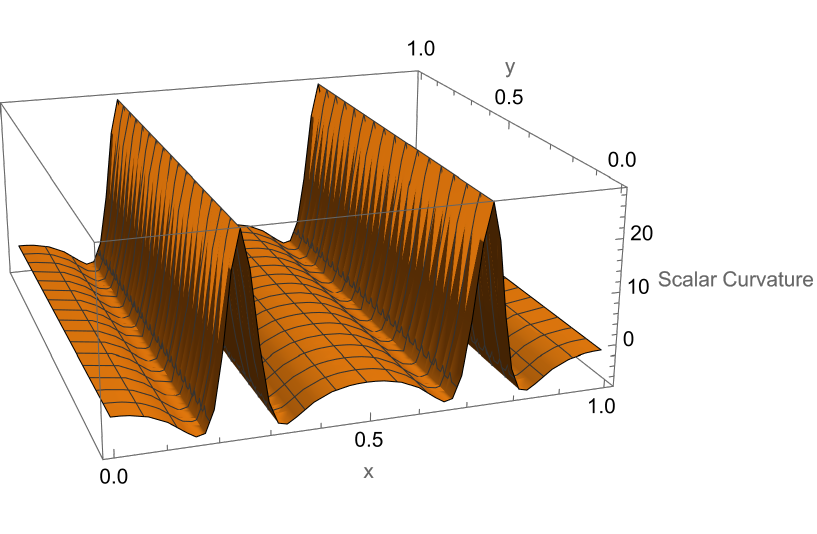}\\[-0.5ex]
    {\small (b) Scalar Curvature \(s=0.1\)}
  \end{minipage}\\[0.5cm]

  \begin{minipage}[b]{0.48\linewidth}
    \centering
    \includegraphics[width=0.8\linewidth]{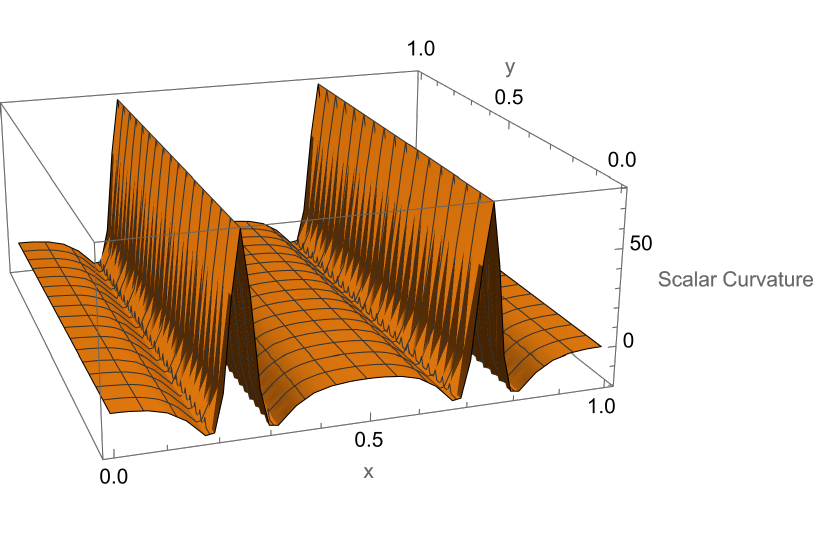}\\[-0.5ex]
    {\small (a) Scalar Curvature \(s=0.12\)}
  \end{minipage}\hfill
  \begin{minipage}[b]{0.48\linewidth}
    \centering
    \includegraphics[width=0.8\linewidth]{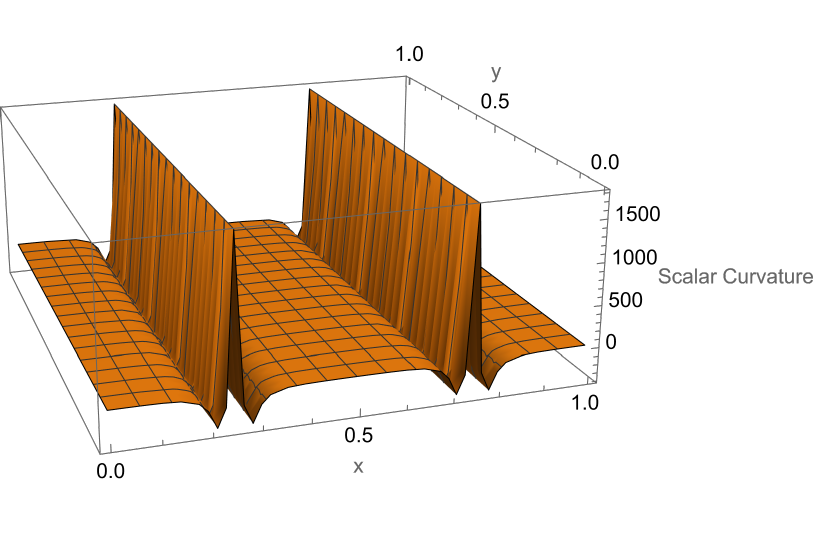}\\[-0.5ex]
    {\small (b) Scalar Curvature \(s=0.15\)}
  \end{minipage}\\[0.5cm]

  \caption{Gauss and single‐particle scalar curvature plots for various values of \(s\).}
  \label{fig:scalar_curvatures}
\end{figure}

From the plots above, we observe that regions of the torus with higher scalar curvature correspond to areas with greater deformation in the density profiles. This correlation highlights the geometric influence on the particle distribution. Next, we examine the impact of $N_{\phi}$ on the density profile by varying $s$ and $\tau$ fixed.

\begin{figure}[htbp]
  \centering

  \begin{minipage}[b]{0.48\linewidth}
    \centering
    \includegraphics[width=0.9\linewidth]{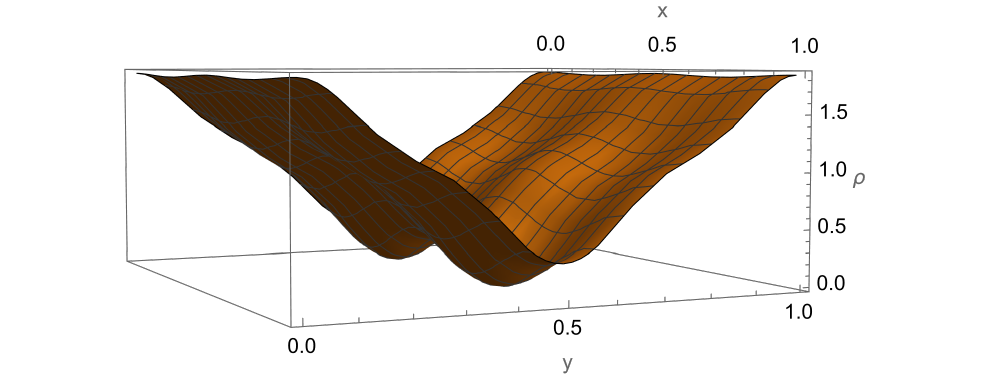}\\[-0.5ex]
    {\small (a) \(N_{\phi}=2\)}
  \end{minipage}\hfill
  \begin{minipage}[b]{0.48\linewidth}
    \centering
    \includegraphics[width=0.9\linewidth]{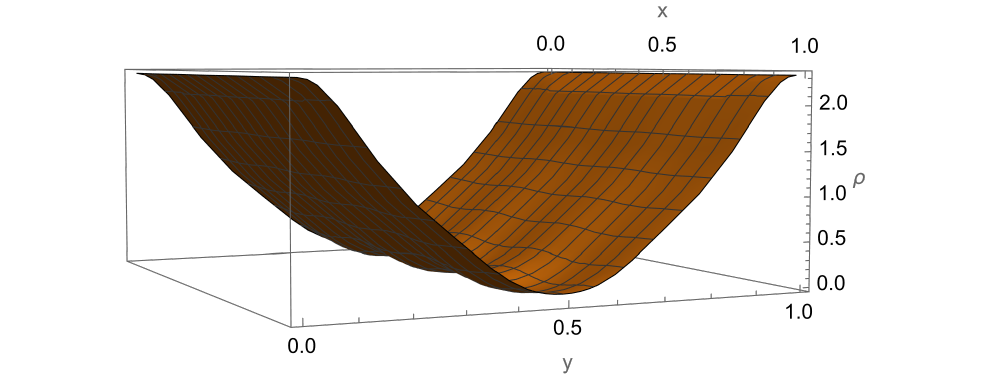}\\[-0.5ex]
    {\small (b) \(N_{\phi}=3\)}
  \end{minipage}\\[0.5cm]

  \begin{minipage}[b]{0.48\linewidth}
    \centering
    \includegraphics[width=0.9\linewidth]{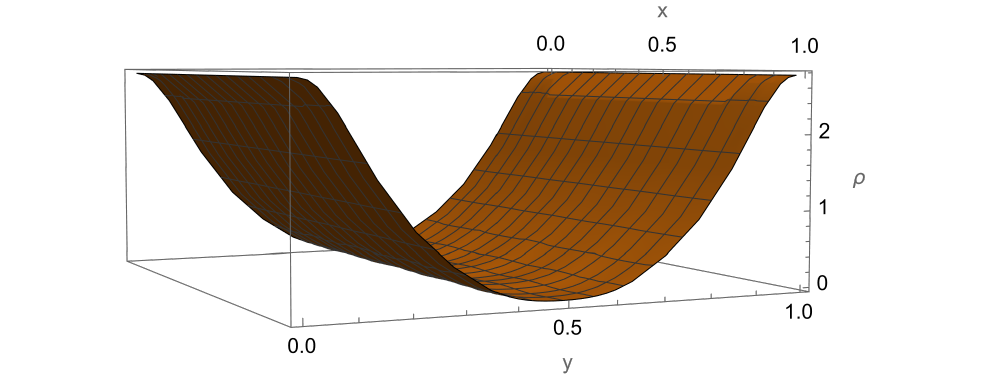}\\[-0.5ex]
    {\small (c) \(N_{\phi}=4\)}
  \end{minipage}\hfill
  \begin{minipage}[b]{0.48\linewidth}
    \centering
    \includegraphics[width=0.9\linewidth]{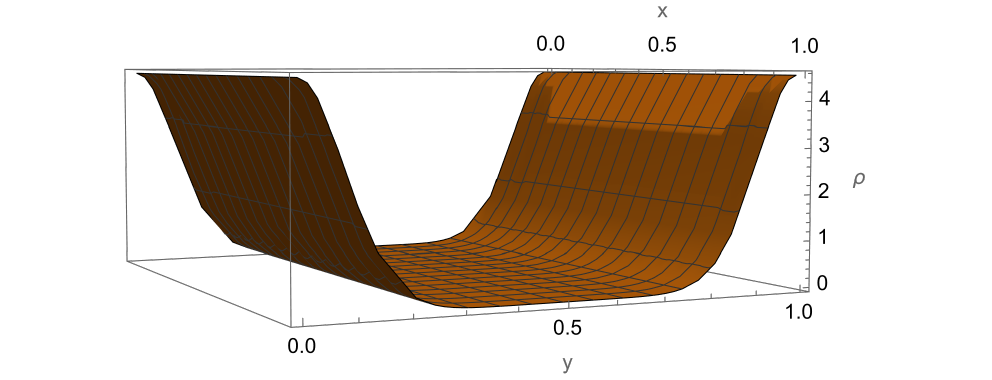}\\[-0.5ex]
    {\small (d) \(N_{\phi}=10\)}
  \end{minipage}

  \caption{Varying \(N_{\phi}\) while keeping \(s=0.1\) and \(\tau=i\).}
  \label{fig:changeNphiSingleparticle}
\end{figure}

\subsection{Integer Quantum Hall Effect and evolution under the gCST}
\label{section:integerquantumhalleffectnonflat}

In this section, we analyze the evolution of the IQHE wavefunction, given in Eq.~\eqref{eq:integerwavefunction}, under the gCST in the non-flat geometry derived in the previous section:
\begin{equation*}
\Psi_{\mathrm{Integer}}(z_{1},\dots,z_{N}) = \vartheta\begin{bmatrix}
\frac{N-1}{2}\\
\frac{N-1}{2}
\end{bmatrix}(Z,\tau)\prod_{1\leq i,j\leq N}\theta_{11}(z_{i}-z_{j},\tau)e^{i\pi \tau N\sum_{j=1}^{N}(y_{j})^2}.
\end{equation*}
Following the approach in the flat geometry case, we separately examine the gCST's effect on the center-of-mass component and the one-particle wavefunctions $\theta_{11}(z_{i}-z_{j},\tau)$. We begin by determining the transformation of each $\theta_{11}(z_{i}-z_{j},\tau)$ under the gCST. Since the previous section established the single-particle state evolution in the non-flat geometry, we apply the same methodology to obtain:
\begin{align*}
U_{s}\theta_{11}(z_{i}-z_{j},\tau) &= e^{s\sin^2{(2\pi(y{i}-y_{j}))}}e^{-2\pi s(y_{i}-y_{j})\sin{4\pi(y_{i}-y_{j})}}\theta_{11}((z_{i}-z_{j})_{is},\tau_s),
\end{align*}
with $(z_{i}-z_{j})_{is}=z_{i}-z_{j}+\frac{\sin{(4\pi(y_{i}-y_{j})is)}}{N}$.
Next, we compute the evolution of the center-of-mass component under the gCST. Recall that the gCST takes the form:
\begin{equation*}
U_{s} = \left(e^{s\sum_{j=1}^{N}\sin^2{2\pi y_{j}}}e^{-2\pi s \sum_{j=1}^{N}y_{j}\sin{4\pi y_{j}}}e^{is\sum_{j=1}^{N}\frac{\sin{4\pi y_{j}}}{N}\frac{\partial}{\partial x_{j}}}e^{-sQ(H)}\right)\bigg|_{t=is}.
\end{equation*}
The transformation proceeds as follows: first, we act on the coordinates with the Hamiltonian vector field $X_H = \sum_{j=1}^{N} \frac{\sin(4\pi y_j)}{N} \frac{\partial}{\partial x_j}$ by applying its flow; then, we act with the function. In analogy with the single-particle case, the flow generated by this vector field is given by:
\begin{equation*}
\phi^{t}_{X_{H}}(\sum_{j=1}^{N}x_{j},\sum_{j=1}^{N}y_{j}) = \left(\sum_{j=1}^{N}x_{j} + \frac{\sin\left(4\pi \left(\sum_{j=1}^{N}y_{j}\right)t\right)}{N_{\phi}}, \sum_{j=1}^{N}y_{j}\right).
\end{equation*}
Applying this to the center-of-mass theta function:
\begin{align*}
e^{t\sum_{j=1}^{N}\frac{\sin{4\pi y_{j}}}{N}\frac{\partial}{\partial x_{j}}} \vartheta\begin{bmatrix}
\frac{N-1}{2}\\
\frac{N-1}{2}
\end{bmatrix}(Z,\tau) &= \vartheta\begin{bmatrix}
\frac{N-1}{2}\\
\frac{N-1}{2}
\end{bmatrix}\left(Z + \frac{\sin{(4\pi Y)}t}{N_{\phi}}, \tau\right),
\end{align*}
where $X=\sum_{j=1}^{N}x_{j}$, $Y=\sum_{j=1}^{N}y_{j}$ and $Z_{is}=Z+\frac{\sin{(4\pi Y)}is}{N_{\phi}}$. Next, the quantum operator's action follows from the spectral theorem:
\begin{equation*}
e^{-sQ(H)}\vartheta\begin{bmatrix}
\frac{N-1}{2}\\
\frac{N-1}{2}
\end{bmatrix}(Z_{is},\tau) = e^{-s\sin^2{\left(2\pi\frac{N-1}{2}\right)}}\vartheta\begin{bmatrix}
\frac{N-1}{2}\\
\frac{N-1}{2}
\end{bmatrix}(Z_{is},\tau).
\end{equation*}
Combining both results, the evolution of the center-of-mass component under the gCST is:
\begin{equation*}
U_{s}\vartheta\begin{bmatrix}
\frac{N-1}{2}\\
\frac{N-1}{2}
\end{bmatrix}(Z,\tau) = e^{s\sin^2{(2\pi Y)}}e^{-2\pi s Y\sin{(4\pi Y)}}e^{-s\sin^2{\left(2\pi\frac{N-1}{2}\right)}}\vartheta\begin{bmatrix}
\frac{N-1}{2}\\
\frac{N-1}{2}
\end{bmatrix}(Z_{is},\tau).
\end{equation*}
Finally, assembling both transformed components, the full evolution of $\Psi_{\mathrm{integer}}$ under the gCST is:
\begin{align*}
U_{s}\Psi_{\mathrm{Integer}} &= e^{s\sin^2{(2\pi Y)}}e^{-2\pi s Y\sin{(4\pi Y)}}e^{-s\sin^2{\left(2\pi\frac{N-1}{2}\right)}}\vartheta\begin{bmatrix}
\frac{N-1}{2}\\
\frac{N-1}{2}
\end{bmatrix}(Z_{is},\tau) \\
&\quad \times \prod_{1\leq i,j\leq N}e^{s\sin^2{(2\pi(y_{i}-y_{j}))}}e^{-2\pi s(y_{i}-y_{j})\sin{4\pi(y_{i}-y_{j})}}\theta_{11}((z_{i}-z_{j})_{is},\tau_s)e^{i\pi \tau N\sum_{j=1}^{N}(y_{j})^2}.
\end{align*}

Now that we have computed the evolution of the wavefunction for the IQHE, we aim to determine the particle density. This was carried out numerically using the expression in Eq.~\eqref{densitysecondquantization}. We investigate the density for the same scenarios as in the single-particle case: fixing $N=2$ and $\tau=i$ while varying $s$, and fixing $s=0.1$ and $\tau=i$ while varying $N$. As before, we restrict $s$ to values that do not cause the Gauss curvature of the torus to diverge. The results for $N=2$ and $\tau=i$ with varying $s$ are presented in Fig.~\ref{fig:change_s_N2_tau_i}.
\begin{figure}[htbp]
  \centering

  \begin{minipage}[b]{0.48\linewidth}
    \centering
    \includegraphics[width=0.7\linewidth]{Figures/ScalarCurvatures=0.pdf}\\[-0.5ex]
    {\small (a) Scalar Curvature \(s=0\)}
  \end{minipage}\hfill
  \begin{minipage}[b]{0.48\linewidth}
    \centering
    \includegraphics[width=0.7\linewidth]{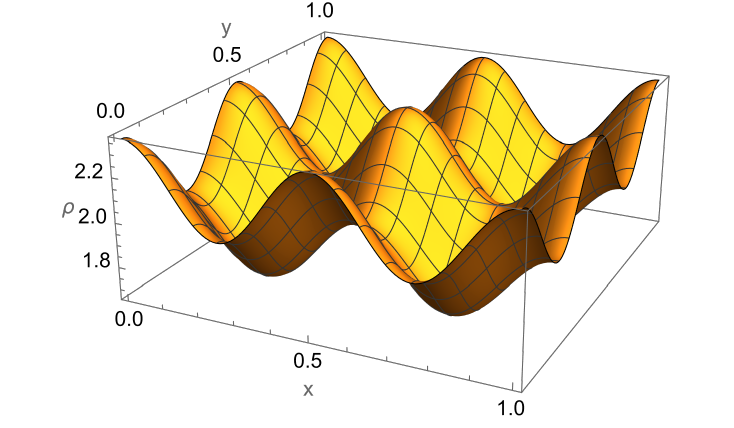}\\[-0.5ex]
    {\small (b) \(s=0\)}
  \end{minipage}\\[0.5cm]

  \begin{minipage}[b]{0.48\linewidth}
    \centering
    \includegraphics[width=0.7\linewidth]{Figures/ScalarCurvatures=0.1.pdf}\\[-0.5ex]
    {\small (c) Scalar Curvature \(s=0.1\)}
  \end{minipage}\hfill
  \begin{minipage}[b]{0.48\linewidth}
    \centering
    \includegraphics[width=0.7\linewidth]{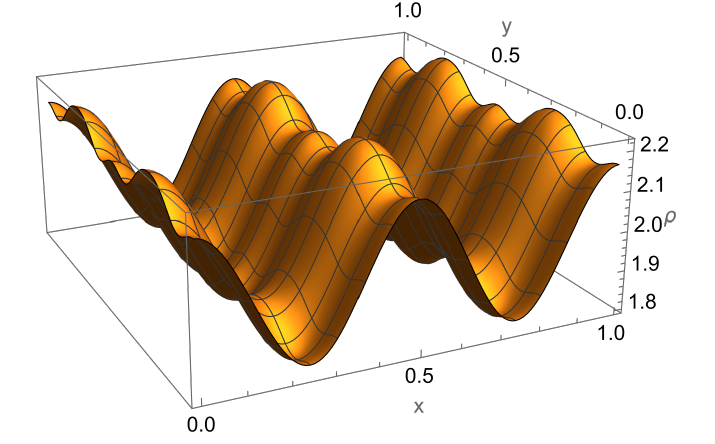}\\[-0.5ex]
    {\small (d) \(s=0.1\)}
  \end{minipage}\\[0.5cm]

  \begin{minipage}[b]{0.48\linewidth}
    \centering
    \includegraphics[width=0.7\linewidth]{Figures/ScalarCurvatures=0.12.pdf}\\[-0.5ex]
    {\small (e) Scalar Curvature \(s=0.12\)}
  \end{minipage}\hfill
  \begin{minipage}[b]{0.48\linewidth}
    \centering
    \includegraphics[width=0.7\linewidth]{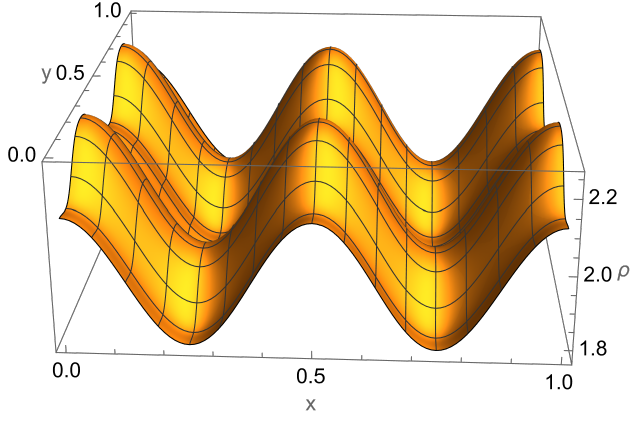}\\[-0.5ex]
    {\small (f) \(s=0.12\)}
  \end{minipage}\\[0.5cm]

  \begin{minipage}[b]{0.48\linewidth}
    \centering
    \includegraphics[width=0.7\linewidth]{Figures/ScalarCurvatures=0.15.pdf}\\[-0.5ex]
    {\small (g) Scalar Curvature \(s=0.15\)}
  \end{minipage}\hfill
  \begin{minipage}[b]{0.48\linewidth}
    \centering
    \includegraphics[width=0.7\linewidth]{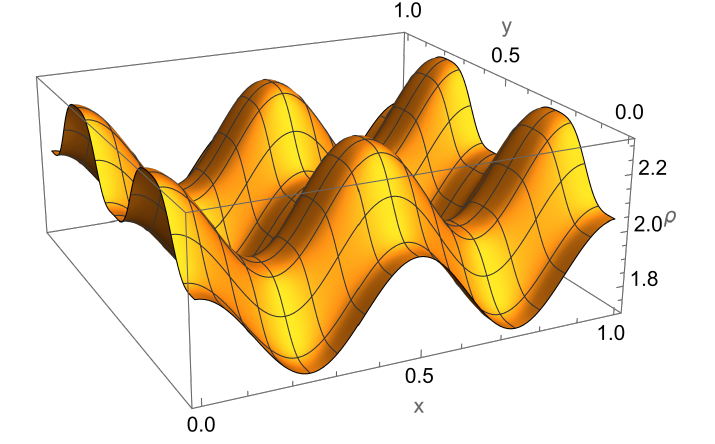}\\[-0.5ex]
    {\small (h) \(s=0.15\)}
  \end{minipage}

  \caption{Varying \(s\) while keeping \(N=2\) and \(\tau=i\).}
  \label{fig:change_s_N2_tau_i}
\end{figure}

As observed previously, the deformation of the density plot is more pronounced at points where the scalar curvature of the torus is higher. Next, we analyze the density plot for $s=0.1$ and $\tau=i$ fixed while varying $N$. The results can be found in Fig.~\ref{fig:changeN_IQHEnonflat}.

\begin{figure}[htbp]
  \centering

  \begin{minipage}[b]{0.32\linewidth}
    \centering
    \includegraphics[width=0.9\linewidth]{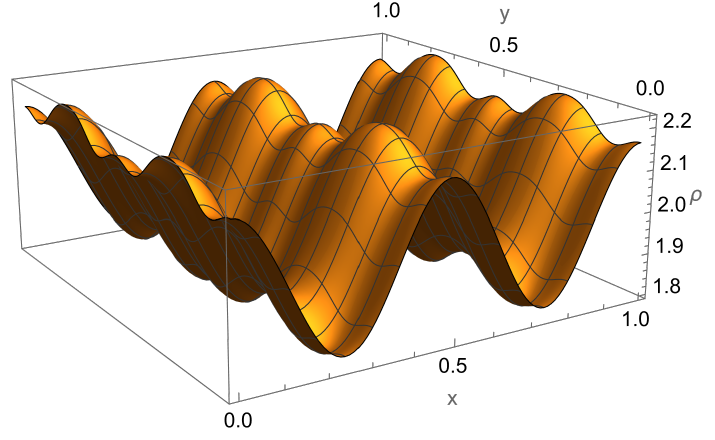}\\[-0.5ex]
    {\small (a) $N=2$}
  \end{minipage}\hfill
  \begin{minipage}[b]{0.32\linewidth}
    \centering
    \includegraphics[width=0.9\linewidth]{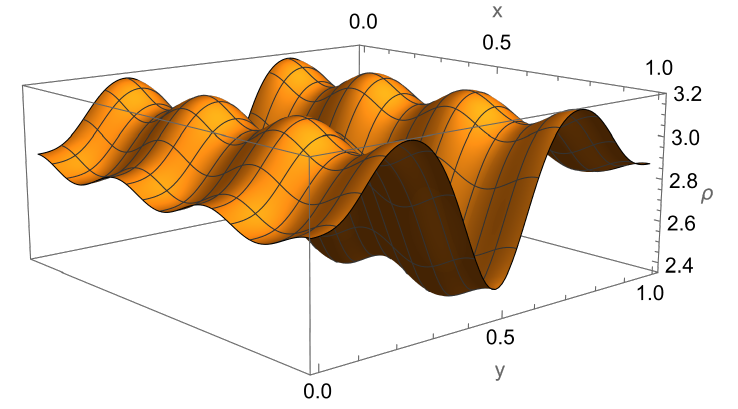}\\[-0.5ex]
    {\small (b) $N=3$}
  \end{minipage}\hfill
  \begin{minipage}[b]{0.32\linewidth}
    \centering
    \includegraphics[width=0.9\linewidth]{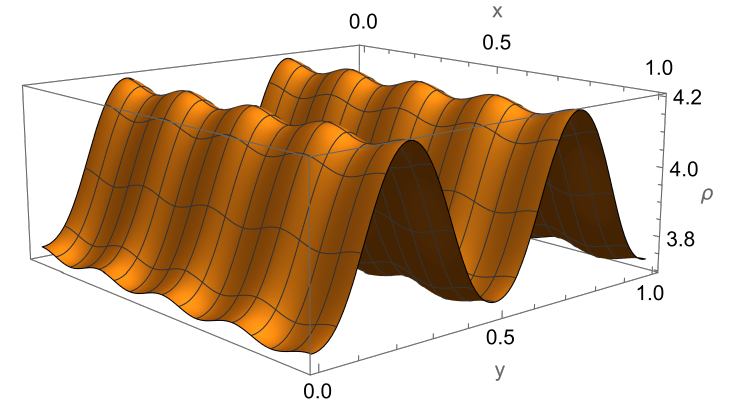}\\[-0.5ex]
    {\small (c) $N=4$}
  \end{minipage}\\[0.5cm]

  \begin{minipage}[b]{0.32\linewidth}
    \centering
    \includegraphics[width=0.9\linewidth]{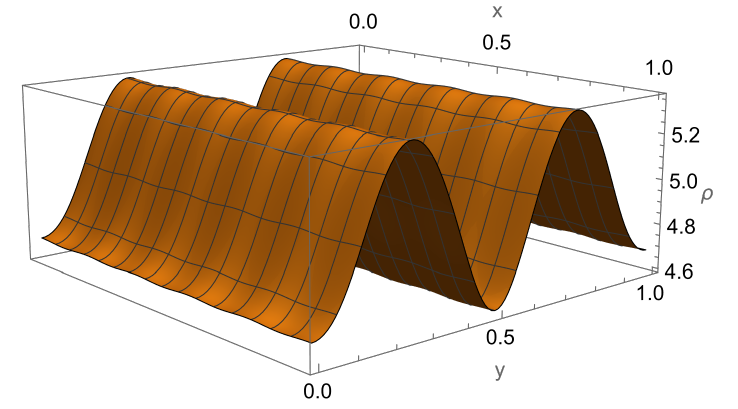}\\[-0.5ex]
    {\small (d) $N=5$}
  \end{minipage}\hfill
  \begin{minipage}[b]{0.32\linewidth}
    \centering
    \includegraphics[width=0.9\linewidth]{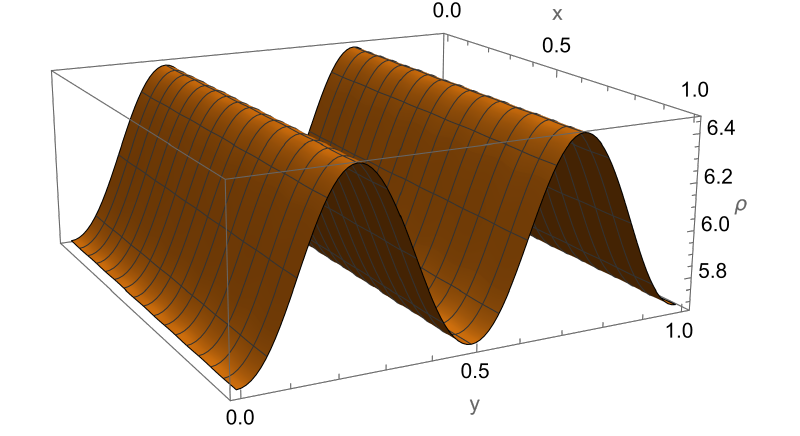}\\[-0.5ex]
    {\small (e) $N=6$}
  \end{minipage}\hfill
  \begin{minipage}[b]{0.32\linewidth}
    \centering
    \includegraphics[width=0.9\linewidth]{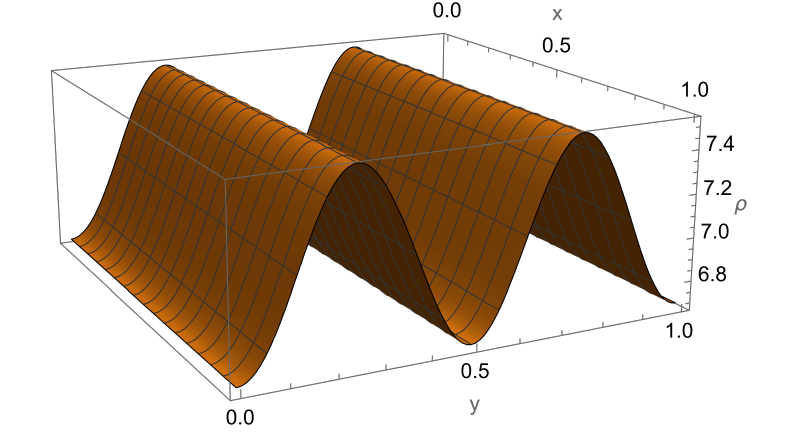}\\[-0.5ex]
    {\small (f) $N=7$}
  \end{minipage}

  \caption{Varying $N$ while keeping $s=0.1$ and $\tau=i$.}
  \label{fig:changeN_IQHEnonflat}
\end{figure}

\subsection{Fractional Quantum Hall Effect and evolution under the gCST}
\label{section:fractionalquantumhalleffectnonflat}

After exploring the IQHE in the non-flat case, we now turn to the FQHE. The wavefunction for the FQHE is given by Eq.~\eqref{eq:Laughlinfractional}:

\begin{equation*}
\Psi_{\mathrm{Laughlin}}(z_{1},\dots,z_{N_{e}})=\vartheta\begin{bmatrix}
\frac{N-1}{2k}+\frac{l}{k}\\
\frac{N-1}{2}
\end{bmatrix}(kZ,k\tau)\prod_{1\leq i,j\leq N_{e}}(\theta_{11}(z_{i}-z_{j},\tau))^{k}e^{i\pi \tau N_{\phi}\sum_{j=1}^{N_{e}}(y_{j})^2},
\end{equation*}

where the filling fraction is $\nu = N_{e}/N_{\phi} = 1/k$. We now apply the non-flat gCST to this wavefunction. As in the flat case, the transformation acts separately on each factor of the Laughlin wavefunction. Since we have already determined how $\theta_{11}(z_{i}-z_{j},\tau)$ evolves in the non-flat IQHE case, it follows that its $k^{th}$ power evolves as:

\begin{equation}
U_{s}(\theta_{11}(z_{i}-z_{j},\tau))^{k} = e^{s k\sin^2{(2\pi(y_{i}-y_{j}))}}e^{-2sk\pi(y_{i}-y_{j})\sin{4\pi(y_{i}-y_{j})}}(\theta_{11}((z_{i}-z_{j})_{is},\tau_s))^{k}. 
\end{equation}

The evolution of the center-of-mass component follows analogously to the IQHE case. We then obtain the following evolution of the FQHE wavefunction under the non-flat gCST:

\begin{align*}
&U_{s}\Psi_{\mathrm{Laughlin}} = e^{s\sin^2{2\pi Y}}e^{-s2\pi Y\sin{4\pi Y}}e^{-s\sin^2{\left(\frac{N_{e}-1}{2k}+\frac{l}{k}\right)}}\vartheta\begin{bmatrix}
\frac{N-1}{2k}+\frac{l}{k}\\
\frac{N-1}{2}
\end{bmatrix}(kZ_{s},k\tau)\cdot\\
&\times\left(\prod_{1\leq i,j\leq N_{e}}e^{sk\sin^2{(2\pi(y_{i}-y_{j}))}}e^{-s2k\pi(y_{i}-y_{j})\sin{4\pi(y_{i}-y_{j})}}(\theta_{11}((z_{i}-z_{j})_{s},\tau_s))^{k}\right)e^{i\pi \tau N_{\phi}\sum_{j=1}^{N_{e}}(y_{j})^2}, 
\end{align*}
with $Z_{is} = Z+\frac{\sin{(4\pi Y)}is}{N_{\phi}}$ and $Y = \sum_{j=1}^{N_e} y_{j}$.

Having determined the evolution of the FQHE wavefunction, we now compute the particle density. As in the IQHE case, this is done numerically using Eq.~\eqref{densitysecondquantization}. We consider the filling fraction $\nu=1/3$ (i.e. $k=3$) while keeping $N_{e}$ and $\tau$ fixed, varying $s$ within the limits that prevent the Gauss curvature of the torus from diverging.
Computational constraints make it challenging to analyze cases with varying particle numbers, so we leave this for future work. The results are presented in Fig.~\ref{fig:changes_FQHEnonflat}.

\begin{figure}[htbp]
  \centering

  \begin{minipage}[b]{0.48\linewidth}
    \centering
    \includegraphics[width=0.7\linewidth]{Figures/ScalarCurvatures=0.pdf}\\[-0.5ex]
    {\small (a) Scalar Curvature \(s=0\)}
  \end{minipage}\hfill
  \begin{minipage}[b]{0.48\linewidth}
    \centering
    \includegraphics[width=0.7\linewidth]{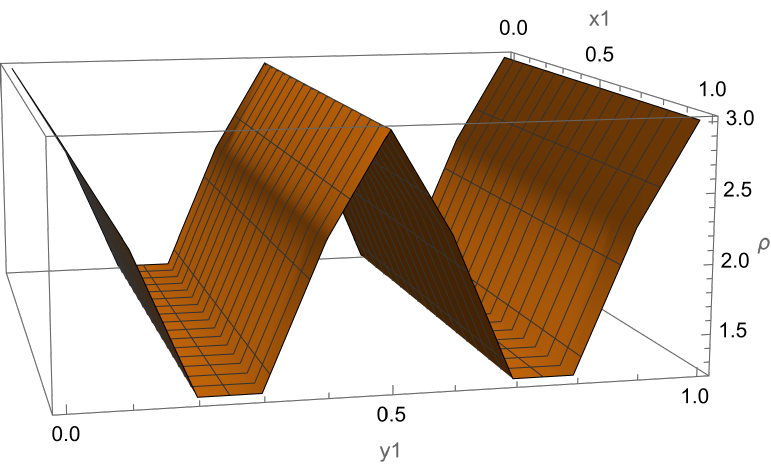}\\[-0.5ex]
    {\small (b) \(s=0\)}
  \end{minipage}\\[0.5cm]

  \begin{minipage}[b]{0.48\linewidth}
    \centering
    \includegraphics[width=0.7\linewidth]{Figures/ScalarCurvatures=0.1.pdf}\\[-0.5ex]
    {\small (c) Scalar Curvature \(s=0.1\)}
  \end{minipage}\hfill
  \begin{minipage}[b]{0.48\linewidth}
    \centering
    \includegraphics[width=0.7\linewidth]{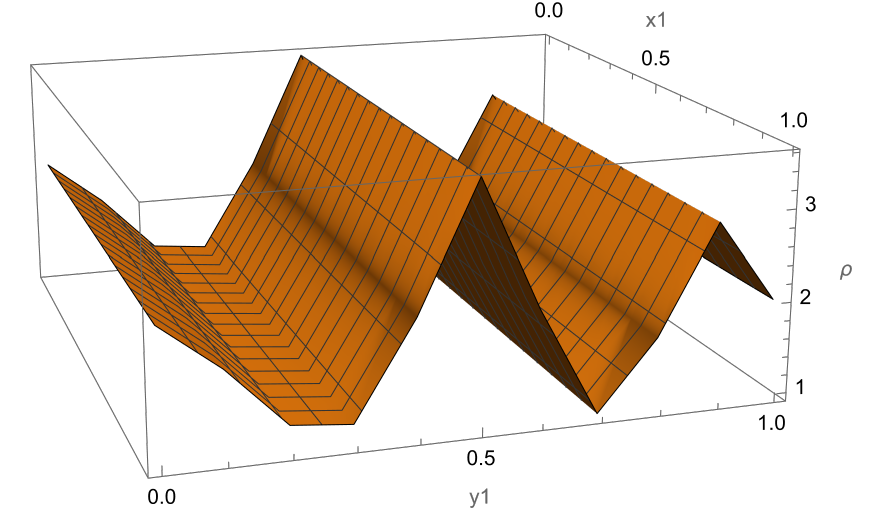}\\[-0.5ex]
    {\small (d) \(s=0.1\)}
  \end{minipage}\\[0.5cm]

  \begin{minipage}[b]{0.48\linewidth}
    \centering
    \includegraphics[width=0.7\linewidth]{Figures/ScalarCurvatures=0.12.pdf}\\[-0.5ex]
    {\small (e) Scalar Curvature \(s=0.12\)}
  \end{minipage}\hfill
  \begin{minipage}[b]{0.48\linewidth}
    \centering
    \includegraphics[width=0.7\linewidth]{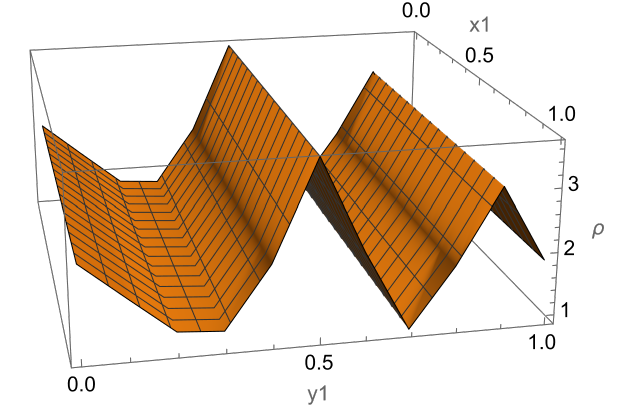}\\[-0.5ex]
    {\small (f) \(s=0.12\)}
  \end{minipage}\\[0.5cm]

  \begin{minipage}[b]{0.48\linewidth}
    \centering
    \includegraphics[width=0.7\linewidth]{Figures/ScalarCurvatures=0.15.pdf}\\[-0.5ex]
    {\small (g) Scalar Curvature \(s=0.15\)}
  \end{minipage}\hfill
  \begin{minipage}[b]{0.48\linewidth}
    \centering
    \includegraphics[width=0.7\linewidth]{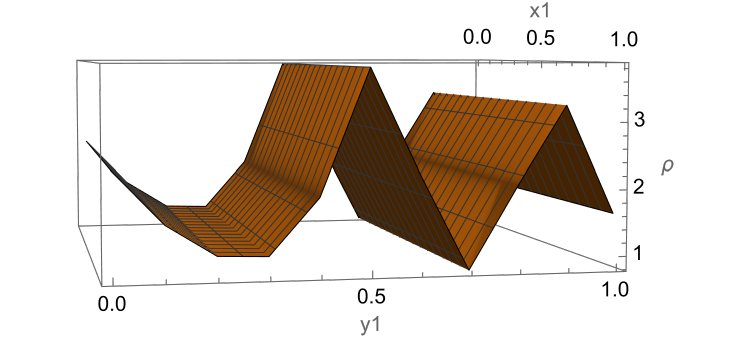}\\[-0.5ex]
    {\small (h) \(s=0.15\)}
  \end{minipage}

  \caption{Varying \(s\) while keeping \(N_{e}=2\) and \(\tau=i\).}
  \label{fig:changes_FQHEnonflat}
\end{figure}

The trends observed in the density plots for one-particle states and the IQHE persist in the FQHE case, reinforcing the validity of our approach.

In both the IQHE and the FQHE cases, the behaviours obtained are consistent with the Wen-Zee effective action for quantum Hall fluids, where the particle density acquires a term proportional to the scalar curvature of the sample, the constant of proportionality being the filling fraction times the shift~\cite{wen:zee:1992, read:2009,read:rezayi:2011}.


\newpage

\begin{appendices}
    
\section{Flat geometries on the plane}
\label{section:flatgeometriesplane}

In section \ref{sec:flat_geometry} we have considered families of flat geometries on the torus which have been generated by the flow in imaginary time of a quadratic in canonical variables (thus non-periodic) Hamiltonian. Note that the resulting complex tori are therefore not biholomorphic to each other, that is they are inequivalent as complex manifolds. We then obtained the Laughlin states for the Hall effect on such flat tori, via the generalized coherent state transform, by deforming the Laughlin states for one such geometry. The Laughlin states obtained via this deformation, coincide with the Laughlin states which have been previously obtained by other methods. Therefore, for the case of flat geometries on the torus, the generalized coherent state transform reproduces, through a process of deformation of a given initial geometry, the correct Laughlin states for the deformed geometries.
This adds plausibility to the results of \cite{matos:mera:mourao:mourao:nunes:2023} and of section \ref{sec:non_flat_geometry}, where, by a similar process of deformation of geometry starting from the round metric on the sphere and the flat metric on the torus, Laughlin states for deformed metrics  with $S^1$-symmetry have been proposed. 

In the present appendix, for completeness, we will show that, for flat metrics on the plane, the generalized coherent state transform also reproduces the known Laughlin states; that is, by starting with the Laughlin states for the standard Euclidean metric and by following an Hamiltonian flow in imaginary time, we obtain the correct Laughlin states for other flat metrics.

Let us consider the symplectic plane $(\R^2,\omega)$, with position and momentum Cartesian coordinates $(x,p)$ and standard symplectic form
$$
\omega = dx \wedge dp.
$$
The Hamiltonian vector field for the Hamiltonian function $H= \frac12 p^2$ is
$$
X_H = p \frac{\partial}{\partial x},
$$
so that the corresponding Hamiltonian flow is 
$$
\varphi^{X_H}_t(x,p) = (x+tp, p),\, t\in \R.
$$

The vertical  polarization leading to the Schr\"odinger quantization of the plane, is generated by the Hamiltonian vector field of $x$,
$$
X_x = -\frac{\partial}{\partial p}.
$$

Translation-invariant K\"ahler structures on the plane can then be obtained by taking the Hamiltonian flow of $X_H$ in imaginary time $\tau = is, s>0$. Indeed, one can define a global holomorphic coordinate
$$
z_s = e^{isX_H} x = x+isp,\, s> 0,
$$
corresponding to a complex structure $J_s, s>0$, so that $(\R^2,\omega, J_s)$ becomes a K\"ahler manifold with Riemannian metric
\begin{equation}\label{flatmetriconplane}
\gamma_s = s^{-1} dx^2 + sdp^2, \, s>0.
\end{equation}

Note that these K\"ahler structures are not rotationally invariant if $s\neq 1.$

These K\"ahler structure has been used in \cite{Johri.Papic.Schmitteckert.Bhatt.Haldane2016} to describe the Hall effect in plane geometries where the fundamental droplets with a given magnetic flux, describing the neighborhood of an a electron, are no longer circular but elliptic.
Note that, in \cite{Johri.Papic.Schmitteckert.Bhatt.Haldane2016} (see their Section 2), this corresponds to the use of the holomorphic coordinate 
$$
z_{(\alpha)}=  \alpha x + \frac{i}{\alpha} p,\ \alpha\in \R,\ \alpha\neq 0. 
$$
This is, of course, equivalent to the global holomorphic coordinate
$$
\alpha^{-1} z_{(\alpha)} = x + \frac{i}{\alpha^2} p,
$$
so that, for $s= {\alpha}^{-2}$, the plane geometry described above, via deformation of the standard Euclidian geometry by Hamiltonian flow in imaginary time, coincides with the one considered in \cite{Johri.Papic.Schmitteckert.Bhatt.Haldane2016}.
For filling fraction $\nu = \frac13$ and $N$ particles, the Laughlin state for the deformed geometry, considered in \cite{Johri.Papic.Schmitteckert.Bhatt.Haldane2016}, with $s= \alpha^{-2}$,
up to an irrelevant overall constant, is
\begin{equation}\label{laughlinhaldaneprobing}
\psi_{\mathrm{Laughlin}} ((z_s)_1, \dots, (z_s)_{N}) = 
\prod_{i<j}^N \left((z_s)_{i}-(z_s)_{j}\right)^3 e^{-\sum_{k} \vert (z_s)_{k}\vert^2/ 4sl_B^2},
\end{equation}
where $l_B$ is the magnetic length. This comes from the one-particle states
\begin{equation}\label{eqflatoneparticle}
\psi_m(z_s) = z_s^m e^{-\vert z_s\vert^2/ 4sl_B^2},\,\,\,  m=0,1, 2, \dots     
\end{equation}
(From \cite{Murayama}, for instance, one easily checks that indeed these give the ground states for the motion of a charged particle under a magnetic field in the geometry given by the flat, but non-isotropic, metric (\ref{flatmetriconplane}).)

We will now show that these states can be described by means of a CST by deformation of the usual Euclidian geometry at $s=1$.
We have the prequantum connection 
$$
\nabla = d + i pdx,
$$
and the prequantum operator
$$
Q_{\mathrm{pre}} (H) = i\nabla_{X_H} + H = i p\frac{\partial}{\partial x} -\frac{p^2}{2}. 
$$
The quantum operator is then
$$
Q(H) = \frac12 (Q_{\mathrm{pre}}(p))^2 = -\frac12 (X_p)^2 = -\frac12 \frac{\partial^2}{\partial x^2},
$$
and the (unitary) CST (\ref{timescst}) at time $t = is$, which corresponds to the classical Segal-Bargmann transform, is given by 
\begin{equation}\label{flatplanecst}
U_{s} = \left(e^{-i \tau Q_{\mathrm{pre}}(H)}\circ e^{i \tau Q(H)}\right)\bigg|_{\tau=is} = 
e^{-\frac{s}{2}p^2} e^{isp\frac{\partial}{\partial x}} e^{\frac{s}{2}\frac{\partial ^2}{\partial x^2}}.
\end{equation}

Note that for $s_1,s_2 >0,$ 
\begin{equation}\label{transitive}
    U_{s_1+s_2} = U_{s_1}\circ U_{s_2}.
\end{equation}

Consider the Hermite polynomials, for $a >0, m\in \mathbb{N}_0$ \cite{Hall00, Ab},
$$
H_m^a(x) = (-1)^m e^{\frac{x^2}{a}} \frac{d^m}{dx^m} e^{-\frac{x^2}{a}} = 
\left( \sqrt{\frac{2}{a}}x-\sqrt{\frac{a}{2}}\frac{d}{dx}\right)^m\cdot 1,
$$
and the Hermite functions 
$$
h_m^a(x) = (-1)^m e^{-\frac{x^2}{a}} H_m^a (x) = \frac{d^m}{dx^m}e^{-\frac{x^2}{a}} =
(-1)^m e^{-\frac{x^2}{a}} \left( \sqrt{\frac{2}{a}}x-\sqrt{\frac{a}{2}}\frac{d}{dx}\right)^m\cdot 1.
$$
Note that $\left\{ h_m^a\right\}_{m\in \mathbb{N}_0}$ is an orthogonal basis of $L^2(\mathbb{R},dx)$.
Using the expression for the heat kernel we obtain that
$$
U_s \left( h_m^a\right) = e^{-\frac{s}{2}p^2} e^{isp\frac{\partial}{\partial x}}\int_{\mathbb{R}} e^{-\frac{(y-x)^2}{2s}} (-1)^m e^{-\frac{y^2}{a}} 
\left( \sqrt{\frac{2}{a}}y-\sqrt{\frac{a}{2}}\frac{d}{dy}\right)^m\cdot 1 \; dy.
$$

\begin{proposition}\label{hermiteprop}For $s>0$, $m\in \mathbb{N}_0$, setting $a=2s$,
$$
    U_s \left(h_m^{2s}\right) = c(m,s)  e^{-\frac{i}{2s}xp}z_s^m e^{-\frac{\vert z_s\vert^2}{4s}} = c(m,s) e^{-\frac{i}{2s}xp} \psi_m (z_s), 
$$
where $c(m,s)$ is a constant, the phase $e^{-\frac{i}{2s}xp}$ is a gauge transformation and we take $l_B=1$ in (\ref{eqflatoneparticle}) for simplicity.
\end{proposition}

\begin{proof}
This follows by straightforward calculation, using integration by parts, noticing that
$$
\left(\sqrt{\frac{2}{a}}y + \sqrt{\frac{a}{2}} \frac{d}{dy}  \right) 
e^{-\frac{2s+a}{2sa}\left(y-\frac{a}{(2s+a)^2}x\right)^2} = 
$$
$$
=
\left(\frac{\sqrt{2a}}{(2s+a)^2} x + \left(\sqrt{\frac{2}{a}}s(2s+a) - (2s+a)^2\right) \frac{d}{dx}   \right) e^{-\frac{2s+a}{2sa}\left(y-\frac{a}{(2s+a)^2}x\right)^2}
 $$
 and by evaluation of the Gaussian integral.
\end{proof}

It is clear that starting, say, at $s=1$, one can generate  the one-particle states for all the deformed non-isotropic flat geometries by applying $U_s$ for an appropriate value of $s$ and using the property of transitivity (\ref{transitive}). Concretely, note that in Proposition \ref{hermiteprop} one has the explicit expression for the map between the Hilbert space of the vertical polarization, $\mathcal{H}_0$ and the Hilbert space of the holomorphic polarization at time $s>0$, $\mathcal{H}_s$, which is most usefully written in the above basis. Of course, one also has maps between the Hilbert spaces of holomorphic one-particle states
$$
U_{s_2}: \mathcal{H}_{s_1}\to \mathcal{H}_{s_1+s_2},
$$
for $s_1, s_2>0$, which are also given directly by the operator (\ref{flatplanecst}) in view of (\ref{transitive}).

We therefore obtain, also for these families of non-isotropic flat geometries on the plane, that Laughlin states can be obtained by applying a CST to the Laughlin state for one such fixed geometry. 
{}

\end{appendices}


%
\section*{Acknowledgments}
B.~M. acknowledges support from the Security and Quantum Information Group (SQIG) in Instituto de Telecomunica\c{c}\~{o}es, Lisbon. This work is funded by FCT/MECI through national funds and when applicable co-funded EU funds under UID/50008: Instituto de Telecomunicações (IT). B.~M. further acknowledges the Scientific Employment Stimulus --- Individual Call (CEEC Individual) --- 2022.05522.CEECIND/CP1716/CT0001, with DOI 10.54499/2022.05522.CEECIND/CP1716/CT0001.  C.P. acknowledges support from the US-Israel Binational Science Foundation (BSF,
No.2018226), Jerusalem, Israel, and the ISRAEL SCIENCE
FOUNDATION (ISF, Grants No. 2307/24, No. 1077/23, and
No. 1916/23). C.P. wants to thank the Alexander Zaks scholarship for its support.
JM and JPN were funded by  FCT/Portugal through the project CAMGSD UID/04459/2025.

\bibliographystyle{unsrt}
\bibliography{biblio}   

\begin{thebibliography}{10}

\bibitem{Woodhouse}
Nicholas Woodhouse.
\newblock {\em Geometric quantization}.
\newblock Oxford University Press, 1992.

\bibitem{donaldson:1999}
Simon Donaldson.
\newblock Symmetric spaces, {K}\"ahler geometry and {H}amiltonian dynamics.
\newblock In {\em Northern {C}alifornia {S}ymplectic {G}eometry {S}eminar}, volume 196 of {\em Amer. Math. Soc. Transl. Ser. 2}, pages 13--33. Amer. Math. Soc., Providence, RI, 1999.

\bibitem{MouraoPimentelComplextime}
J.~M. Mourão and J.~P. Nunes.
\newblock On complexified analytic {H}amiltonian flows and geodesics on the space of {K}ahler metrics, 2013.

\bibitem{MouraoPimentelGCST}
W.~D. Kirwin, J.~M. Mourão, and J.~P. Nunes.
\newblock Complex time evolution in geometric quantization and generalized coherent state transforms, 2012.

\bibitem{matos:mera:mourao:mourao:nunes:2023}
Gabriel Matos, Bruno Mera, Jos{\'e}~M. Mour{\~a}o, Paulo~D. Mour{\~a}o, and Jo{\~a}o~P. Nunes.
\newblock {Laughlin States Change Under Large Geometry Deformations and Imaginary Time Hamiltonian Dynamics}.
\newblock {\em Communications in Mathematical Physics}, 399(3):2045--2070, 2023.

\bibitem{Haldane2}
F.~D.~M. Haldane.
\newblock Fractional quantization of the {H}all effect: A hierarchy of incompressible quantum fluid states.
\newblock {\em Phys. Rev. Lett.}, 51(7):605--608, 1983.

\bibitem{Haldane-Rezayi}
E.H.~Rezayi F.M.D.~Haldane.
\newblock Periodic {L}aughlin-{J}astrow wave functions for the fractional quantized {H}all effect.
\newblock {\em Physical Review B}, 31:2529(R), 1985.

\bibitem{Murayama}
H.~Murayama.
\newblock Lecture notes for 221{A}, {UC} {B}erkeley.

\bibitem{Johri.Papic.Schmitteckert.Bhatt.Haldane2016}
S.~J. Johri, Z.~Papi\'c, P.~Schmitteckert, R.~N. Bhatt, and F.~D.~M. Haldane.
\newblock Probing the geometry of the {L}aughlin state.
\newblock {\em New Journal of Physics}, 18:025011, 2016.

\bibitem{BrianHall}
Brian~C. Hall.
\newblock {\em Quantum Theory for Mathematicians}.
\newblock Springer New York, 2013.

\bibitem{kempf}
George Kempf.
\newblock {\em Complex algebraic varieties and theta functions}.
\newblock Springer-Verlag, 1991.

\bibitem{BaierMouraoNunes}
T.Baier, J.~M. Mourão, and J.~P. Nunes.
\newblock Quantization of abelian varieties: Distributional sections and the transition from kähler to real polarizations.
\newblock {\em Journal of Functional Analysis}, 258:3388--3412, 2010.

\bibitem{LargeLimitKlevtsov}
Semyon Klevtsov.
\newblock Geometry and large {N} limits in {L}aughlin states, 2016.

\bibitem{wen:niu:1990}
X.~G. Wen and Q.~Niu.
\newblock {Ground-state degeneracy of the fractional quantum Hall states in the presence of a random potential and on high-genus Riemann surfaces}.
\newblock {\em Phys. Rev. B}, 41:9377--9396, May 1990.

\bibitem{PhysRevB.77.155308}
E.~J. Bergholtz and A.~Karlhede.
\newblock Quantum {H}all system in {T}ao-{T}houless limit.
\newblock {\em Phys. Rev. B}, 77:155308, Apr 2008.

\bibitem{hansson2009taothoulessrevisited}
T.~H. Hansson and A.~Karlhede.
\newblock {T}ao-{T}houless revisited, 2009.

\bibitem{Zhou_Nussinov_Seidel}
A.~Seidel Z.~Zhou, Z.~Nussinov.
\newblock Heat equation approach to geometric changes of the torus {L}aughlin state.
\newblock {\em Physical Review B}, 87:115103, 2013.

\bibitem{Mumford83}
David Mumford.
\newblock {\em Tata Lectures on Theta I}.
\newblock Progress in Mathematics, Vol.28. 1983.

\bibitem{wen:zee:1992}
Xiao-Gang Wen and A.~Zee.
\newblock {Shift and spin vector: New topological quantum numbers for the Hall fluids}.
\newblock {\em Physical Review Letters}, 69(6):953--956, 1992.

\bibitem{read:2009}
N.~Read.
\newblock {Non-Abelian adiabatic statistics and Hall viscosity in quantum Hall states and $p_x + ip_y$ paired superfluids}.
\newblock {\em Physical Review B}, 79(4):045308, 2009.

\bibitem{read:rezayi:2011}
{Hall viscosity, orbital spin, and geometry: Paired superfluids and quantum Hall systems}, author = {Read, N. and Rezayi, E. H.}
\newblock {\em Physical Review B}, 84(8):085316, 2011.

\bibitem{Hall00}
Brian Hall.
\newblock {\em Holomorphic methods in analysis and mathematical physics}.
\newblock in First Summer School in Analysis and Mathematical Physics, S. P\'erez-Esteva and C. Villegas-Blas, Eds. 2000.

\bibitem{Ab}
G.T.F Abreu.
\newblock Closed-form correlation functions of generalized hermite wavelets.
\newblock {\em IEEE Trans. Signal Process.}, 54:2258--2262, 2005.

\end{thebibliography}


\end{document}